\documentclass[10pt,twocolumn]{article}

\usepackage[margin=0.75in,columnsep=0.25in]{geometry}
\usepackage{mathptmx} 
\usepackage{amsmath,amssymb,amsthm}
\usepackage{graphicx}
\usepackage{booktabs}
\usepackage{makecell}
\usepackage{microtype}
\usepackage{xcolor}
\usepackage[hidelinks]{hyperref}
\usepackage[capitalise]{cleveref}
\usepackage{enumitem}
\setlist{leftmargin=*,itemsep=1pt,topsep=2pt,parsep=0pt}

\graphicspath{{figs/}}

\newtheorem{theorem}{Theorem}
\newtheorem{lemma}{Lemma}
\newtheorem{proposition}{Proposition}

\newcommand{\sys}{TEMPO}
\newcommand{\nstar}{n^{*}}
\title{\bf \sys{}: Makespan-Aware Expert-Parallel Load Balancing\\
Across Memory- and Compute-Bound Regimes}
\author{Jie Li, Chenxin Jia, Jinliang Shen, Cunzhuang Liu,
Ruiyi Ding, Jianwen Xian, Kang He,
Chengru Song\thanks{Corresponding author.}\\[3pt]
KlingAI Research\\[2pt]
\normalsize\texttt{lijie@klingai.com, songchengru@klingai.com}}
\date{}

\begin{document}
\maketitle

\begin{abstract}
In expert-parallel (EP) Mixture-of-Experts serving, every layer ends at a
synchronization point: the time a batch spends in the MoE block is the
time of the \emph{slowest} GPU. Production dispatchers balance
proxies---token counts (EPLB, LPLB, UltraEP) or activated-expert counts
(METRO)---implicitly assuming expert time is linear in one of the two.
Measurements on two generations of datacenter GPUs\footnote{GPU model identifiers, microarchitecture names, and vendor specifications are anonymized in this version. The older 8-GPU node is referred to as Testbed~A, the newer-generation nodes as Testbed~B.} show it is neither, for two hardware
reasons a token count cannot express: below
$\nstar\!\approx\!156$--$168$ tokens per expert, streaming the expert's
weights from HBM dominates---cost attaches to \emph{activated replicas},
not tokens---and above it, grouped GEMM rounds every expert's tokens up
to 128-token $M$-tiles, so \emph{splitting} an expert manufactures
padded compute. One max-affine profile $t=\max(a+bG,\,c+\beta N)$
prices both (the flat term is the first step of the tile staircase; a
single extra parameter prices the rest). Realistic decode batches hold
hot experts in the linear regime and cold experts in the flat regime
\emph{simultaneously}; on recorded batches, the dispatches the deployed
proxy policies actually produce differ by $1.4$--$1.6\times$ in modeled
block time (p95 up to $1.7\times$), and \emph{which} proxy pays flips
with the regime. We formalize per-batch dispatch as a
fixed-charge makespan problem---NP-hard on two fully replicated GPUs,
polynomial in each degenerate limit, with an additive approximation
guarantee under full replication---and present \sys{}, a
makespan-aware dispatcher that solves it in milliseconds off the
critical path; its SGLang integration runs the solver out-of-process
and fuses dispatch with count collection into one in-graph kernel. On a
calibrated phase diagram, anchored by an 8-GPU Testbed~A wall-clock
microbenchmark, \sys{} stays within 1\% of the best fixed baseline at
every evaluated grid point and wins by up to 15.5\% where regimes mix
(model-scored; the largest wins are EP32--64 extrapolations).
End-to-end on Testbed~B, two flagship models bracket the predicted win
region: Qwen3-235B (inside) gains $4$--$6\%$ throughput on
long-context and decode-heavy traffic, cuts p99 inter-token latency by
${\sim}15.6\%$, and gains $4$--$7\%$ on 2-node EP16 with a
topology-aware source split; DeepSeek-V3 (outside,
communication-dominated) returns only mechanism cost for every
adaptive policy tested. \sys{} also outperforms SGLang's shipped
token-LP dispatcher by a wide margin, though a like-for-like port
attributes most of that gap to integration architecture rather than
the objective. A phase diagram, not a
universal win, is the paper's claim: it predicts both outcomes before
deployment.
\end{abstract}

\section{Introduction}\label{sec:intro}

Mixture-of-Experts (MoE) models dominate the quality-per-FLOP frontier of
large language models~\cite{shazeer2017moe,fedus2022switch,deepseekv3},
and expert parallelism (EP) is the standard way to serve them: experts are
sharded across GPUs, and each layer performs a dispatch all-to-all,
per-expert grouped GEMMs, and a combine all-to-all~\cite{gshard,deepep}.
Both collectives synchronize the EP group, so the MoE block of every layer
costs exactly the \emph{makespan}---the maximum per-GPU time---regardless
of how idle the other GPUs are. Balancing this makespan is the
load-balancing problem this paper addresses.

Production systems attack it in two layers. A \emph{placement} layer
(e.g., DeepSeek's EPLB~\cite{eplb}) replicates hot experts and re-shuffles
them across GPUs on a minutes-scale horizon, using long-run average load.
A \emph{dispatch} layer decides, for each batch, how the tokens of each
expert are split across its replicas. Because batch-to-batch routing
fluctuates far more than the average (the bottleneck GPU of a given batch
is routinely 1.5--2$\times$ the mean), the dispatch layer is where the
remaining imbalance lives, and it is the layer we study.

Existing dispatchers optimize proxies. Token-balancing methods---LPLB's
linear program~\cite{lplb}, UltraEP's exact quotas~\cite{ultraep}, and the
uniform replica split EPLB defaults to---minimize the maximum number of
\emph{tokens} per GPU. Activation-balancing methods---most recently
METRO~\cite{metro}---minimize the maximum number of \emph{activated
experts} per GPU. Each proxy encodes an implicit linearity assumption:
token balancing assumes time $\propto$ tokens; activation balancing
assumes time $\propto$ expert count.

Neither assumption survives measurement. On Testbed~A GPUs, per-expert FFN time
under DeepGEMM fp8 grouped GEMM~\cite{deepgemm} is flat in token count up
to an inflection $\nstar\!\approx\!156$--$168$ tokens per expert (loading
expert weights dominates) and linear beyond it
(\cref{fig:idea}a, \cref{fig:cost}). The per-GPU cost is therefore
two-regime:
\begin{equation}
  t_g \;=\; \max\bigl(a + b\,G_g,\;\; c + \beta\,N_g\bigr),
  \label{eq:cost}
\end{equation}

\begin{figure*}[t]
\centering
\includegraphics[width=0.82\textwidth]{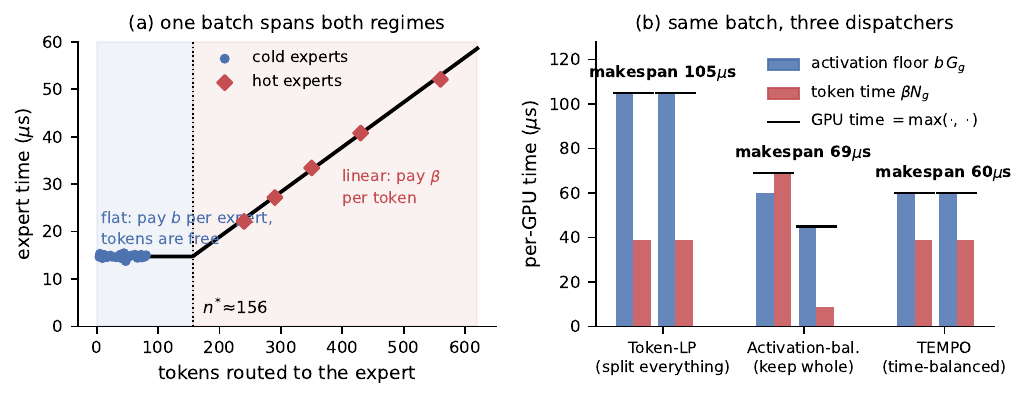}
\caption{The idea in one picture. (a) Measured per-expert time (DSv3
shape, fp8): flat below $\nstar$ (weight loading dominates; extra
tokens free), linear above (every token costs $\beta$). Markers are
the experts of \emph{one} decode batch: both regimes are active
simultaneously. (b) A toy batch (one hot expert, 600 tokens; six cold,
30 each; two GPUs, full replication) dispatched three ways; each GPU's
time is $\max(b\,G_g,\ \beta N_g)$. Token balancing splits everything
and pays 7 floors per GPU ($105\,\mu$s); activation balancing
overloads the hot GPU ($69\,\mu$s); balancing modeled \emph{time}
splits only the hot expert ($60\,\mu$s, $1.75\times$ better than
token-LP).}
\label{fig:idea}
\end{figure*}
where $G_g$ is the number of activated (expert, replica) pairs on GPU $g$
and $N_g$ its token count. The fitted activation floor $b$ spans
$1.74\,\mu$s (Qwen3-30B shape) to $14.8\,\mu$s (DeepSeek-V3 shape)---an
$8.5\times$ spread that grows with expert size---with mean fit error
4--8\% (\cref{tab:params}). Both parameters have direct hardware
readings. The floor $b$ is the cost of streaming one more replica's
weights from HBM: a memory-roofline term that attaches to
\emph{instantiations}, not tokens, so no token count can see it. And
above the inflection, compute itself is not linear in tokens: grouped
GEMM pads every expert's tokens to 128-token $M$-tiles
(\cref{eq:tilecost}), so fragmenting one expert across $k$ replicas
rounds up $k$ times---token balancing does not merely rearrange work,
it \emph{manufactures} padded tiles. The two effects are one mechanism
at two scales: the flat regime is the first step of the tile
staircase, where the step height is dominated by the weight load.

Two consequences follow. First, each proxy is \emph{systematically} wrong
in the other's home regime: on recorded batches---same batch, same
placement, varying only the split---the dispatches the deployed proxy
policies produce differ by $1.2$--$1.6\times$ in modeled block time
(p95 $1.7\times$), each proxy paying $30$--$70\%$ away from its home
regime (\cref{sec:feasible}).
Second---the pivotal empirical fact---the regimes are not separate
operating points a deployment can be tuned for: in real routing traces
at decode batch sizes, \textbf{92--100\% of batches contain both
regimes at once}---the few hot experts sit deep in the linear region
carrying ${\sim}91\%$ of tokens while roughly half of the activated
experts remain in the flat region (\cref{fig:regime}). Token balancing
fragments cold experts and pays hidden activation floors; activation
balancing piles tokens onto hot replicas.

We argue the correct objective is the makespan of the \emph{measured} cost
function \eqref{eq:cost}, and we build \sys{} (Time-modeled
Expert-Parallel Optimization) around it:

\begin{itemize}
\item \textbf{Concept: a cost model that prices what proxies cannot,
  and a phase diagram built on it (\cref{sec:phase}).} The model's two
  non-linear terms are the two hardware effects invisible to token
  counts---memory-bound weight streaming (the floor $b$) and $M$-tile
  padding (a one-parameter staircase term)---measured black-box in ten
  minutes on the deployed kernel. On top of it we formalize per-batch
  dispatch as a fixed-charge makespan problem: NP-hard already with
  two GPUs and full replication, yet polynomial in each degenerate
  limit ($b\!\to\!0$ is the token LP; $\beta\!\to\!0$ a
  semi-matching)---the hardness, like the systems problem, lives in
  the regime interaction. Sweeping the calibrated model yields a phase
  diagram whose winning fixed policy flips along an analytically
  predictable boundary; no fixed proxy wins everywhere.
\item \textbf{Algorithm: \texttt{tempo\_fast} (\cref{sec:design}).}
  Cost-aware seeding, augmenting-chain activation rebalancing (the
  flat-regime move), bottleneck local search with partial migrations
  (the linear-regime move), and an ensemble that scores the token-LP
  and a round-robin certificate under the same model: within 1.02\% of
  a 10-second MILP at ${\sim}2$\,ms, never worse than either classical
  proxy beyond 1\% by construction \emph{under the model}, with an
  additive $\mathrm{OPT}+\max(b,\beta n_{\max})$ guarantee under full
  replication (\cref{thm:rr}).
\item \textbf{Methodology: calibrated simulation anchored by wall
  clock (\cref{sec:eval}).} Headline numbers come from a simulator
  driven by the measured model, and we close the loop: an
  8-GPU Testbed~A EP microbenchmark validates transfer (93\% pairwise
  ranking agreement; 2.2--5.5\,pp gain-transfer error after black-box
  recalibration) and exposes model limits honestly (\cref{sec:limits}).
\item \textbf{Deployability: SGLang integration
  (\cref{sec:integration}).} One fused CUDA-graph-resident kernel
  (probabilistic dispatch + count collection); the solver runs
  out-of-process and publishes tables race-safely. Zero kernels and
  zero collectives on the critical path, verified by a no-op variant
  that ties static within noise (\cref{sec:serving-proxy}). This
  architecture delivers 1.4--1.7$\times$ the throughput of SGLang's
  shipped LPLB dispatcher on Qwen3-235B-FP8 (\cref{sec:serving235b});
  a like-for-like port of token-LP into our worker shows most of that
  gap is architectural, the time model's residual edge being
  stability at the compute-bound point (\cref{sec:serving-drift}).
\item \textbf{Topology: multi-node dispatch needs locality
  (\cref{sec:comm-topo,sec:serving-ep16}).} On 2-node EP16 a flat
  table regresses $3.5\%$ below static at the all-to-all-bound point;
  a two-stage split (makespan solve, then same-node-first
  source-to-replica pairing) preserves per-GPU loads, minimizes
  inter-node traffic, and flips it to $+4\%$ ($+7.2\%$ on Qwen3-235B
  end-to-end), with no kernel changes. A staleness sweep (drift
  $8/16/32$) shows the win region is also a band in time: too little
  drift and there is nothing to fix, too much and only placement can.
\end{itemize}

We are explicit about scope. On small-expert shapes
($b\!\approx\!1.7\,\mu$s) policy choice barely moves wall-clock time
and our DeepSeek-V2-Lite serving runs confirm adaptive dispatch is a
wash there; the win region is flagship shapes (DSv3-class experts,
fp8, decode) where the activation floor is 5--8$\times$ larger. A
third cost dimension the $(G,N)$ model cannot express (kernel
preference for uniform per-slot loads) decides near-ties
(\cref{sec:limits}).

\section{Background and Motivation}\label{sec:motivation}

\subsection{The EP synchronization structure}
Per MoE layer: router $\to$ dispatch all-to-all $\to$ grouped expert GEMMs
(gate/up, activation, down) $\to$ combine all-to-all. Decode servers run
this under CUDA graphs with fused or fixed-shape pre/post kernels, so
the load-dependent term is the grouped GEMM (the all-to-alls scale
with routed bytes, absorbed by the communication extension of
\cref{sec:comm-model}); both all-to-alls synchronize the EP group, so
per-layer MoE time is $\max_g(\text{per-GPU time})$, multiplied across
dozens of layers. The slack is large: the per-batch bottleneck GPU
routinely runs $1.5$--$2\times$ the mean, and that gap between
``balanced on average'' (placement's job) and ``balanced this batch''
(dispatch's job) is where the recoverable time lives.

\subsection{Measured expert cost: weight streaming and tile padding,
one max-affine model}
We benchmark DeepGEMM fp8 masked grouped GEMM (the kernel SGLang/DeepEP
use in decode) across $(G,\text{tokens-per-expert})$ grids for three
expert shapes, with CUDA-graph timing, weight-copy rotation to defeat L2
reuse, and fixed \texttt{expected\_m} to freeze the JIT tuner
(\cref{sec:calib-details}). \cref{fig:cost}(a) shows per-expert time: flat
until $\nstar\!\approx\!156$ tokens, then linear. The two-piece max-affine
model \eqref{eq:cost} fits with 4--8\% mean error (\cref{tab:params});
bf16 loop implementations shift the inflection
($\nstar\!\approx\!331$) but keep the shape.

\begin{figure}[t]
\centering
\includegraphics[width=\linewidth]{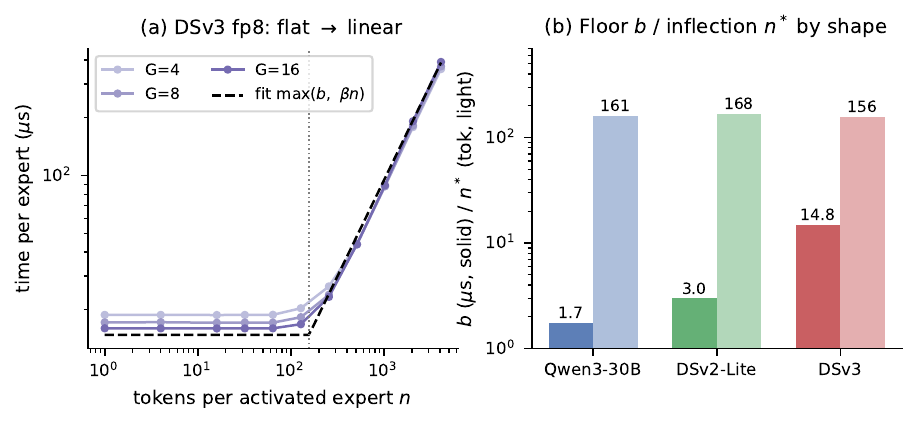}
\caption{(a) DSv3 fp8 grouped GEMM per-expert time: flat below
$\nstar\!\approx\!156$ tokens, linear above; dashed line is the fitted
max-affine model. (b) Activation floor $b$ and inflection $\nstar$ across
expert shapes; $b$ grows $8.5\times$ from Qwen3 to DSv3.}
\label{fig:cost}
\end{figure}

The activation floor $b$ is the physically meaningful quantity: the cost
of touching one more expert's weights on a GPU. A roofline check makes
this concrete: an expert FFN holds $3\,d_{\mathrm{model}}
d_{\mathrm{ff}}$ parameters ($44.0$\,MB for DSv3, $8.7$\,MB for
DSv2-Lite, $4.7$\,MB for Qwen3-30B at fp8), and dividing by Testbed~A's peak HBM bandwidth predicts floors of $13.1/2.6/1.4\,\mu$s.
The fitted $b$ ($14.8/3.0/1.7\,\mu$s, \cref{tab:params}) sits at
$1.13$--$1.24\times$ that bound---the floor \emph{is} weight
streaming, at $80$--$88\%$ of peak bandwidth. (The bf16 loop kernel
fits $45\,\mu$s against a $26\,\mu$s roofline: the two-regime shape
survives, the constant reflects the less efficient implementation.)
Fragmentation under token balancing is charged twice by the
hardware---once in weight streaming (each fragment re-loads the
expert) and once in tile padding (each fragment rounds up
separately)---and $b$ sets the price of the first and dominant charge.
This is what determines where each classical proxy fails.

\paragraph{The newer testbed and the full tile staircase.} Recalibrating on Testbed~B
(fp8 masked DeepGEMM, CUDA-graph timing) reproduces the two-regime
shape (\cref{fig:costtb}) and resolves the underlying mechanism at
fine grain: grouped GEMM processes tokens in $M$-tiles of 128, so cost
is a \emph{staircase} in tokens-per-expert---flat within $[1,128]$, a
jump at 129, flat again to 256 (\cref{fig:tile}). The flat regime of
\eqref{eq:cost} is precisely the first step of this staircase (weight
load dominates the step height); the later steps are sub-linear
($+11$--$24\%$ per boundary; extra tiles of the same expert reuse its
weights from L2), so a single extra parameter captures them:
\begin{equation}
  t = \max\bigl(a + bG + b_2 (T - G),\; c + \beta N\bigr),
  \qquad T = \textstyle\sum_e \lceil n_e / 128 \rceil,
  \label{eq:tilecost}
\end{equation}
with $b_2 \approx b/3$ (Qwen3-235B: $b{=}3.95$, $b_2{=}1.23$;
DSv3: $b{=}8.32$, $b_2{=}2.65\,\mu$s). The tile term cuts model error
in the multi-tile region from ${\sim}10\%$ (worst 25\%) to $2.0\%$ and
leaves the decode region ($n_e\le128$, already priced by $bG$)
untouched; the solver supports it by tracking per-GPU tile counts and
snapping partial splits to tile boundaries (\cref{sec:solver}). On
decode the tile-aware and tile-blind solutions coincide (the floor
\emph{is} the staircase's only active step); at prefill scale the term
earns its keep---on recorded routing rescaled to $4096$--$8192$ tokens
per GPU, tile-aware search changes the tile-blind solution in
$97$--$99\%$ of batches and cuts modeled block time by a median
$4.5$--$6.0\%$ (p95 $9.4\%$), while the uniform replica split
manufactures $3$--$10\%$ extra padded tiles---vanishing by $16{,}384$
tokens per GPU where the linear term dominates. Model-scored numbers
(\cref{sec:transfer}); we enable $b_2$ for prefill-heavy studies and
leave it off for decode.

\begin{figure}[t]
\centering
\includegraphics[width=\linewidth]{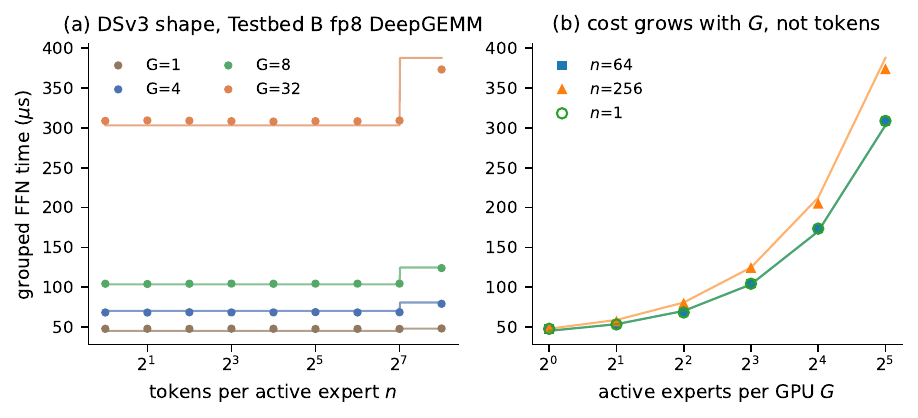}
\caption{Testbed~B fp8 masked DeepGEMM, DSv3 shape. (a) Cost vs.\
tokens-per-expert at several activation counts $G$; lines are the
tile-aware fit \eqref{eq:tilecost}. (b) At fixed tokens-per-expert,
cost is driven by $G$: the $n{=}1$ and $n{=}64$ curves coincide.}
\label{fig:costtb}
\end{figure}

\begin{figure}[t]
\centering
\includegraphics[width=\linewidth]{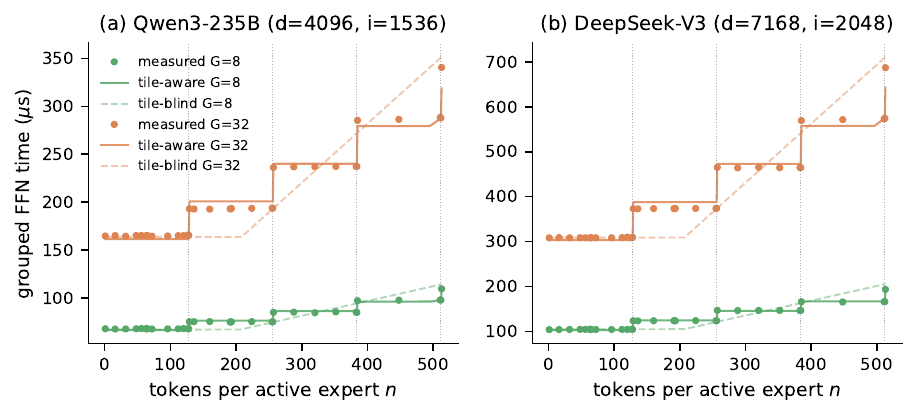}
\caption{Tile staircase, measured with a fine sweep around 128-token
boundaries (dotted verticals). The tile-blind max-affine model (dashed)
misprices the plateaus by up to 25\%; the one-parameter tile term
(solid) tracks the steps at 2\% mean error.}
\label{fig:tile}
\end{figure}

\subsection{What proxy mispricing costs on real batches}
\label{sec:feasible}
Two views, from coarse to operational. On the raw measurement grid
(\cref{fig:proxy}), iso-token sets (same $N$, different $G$) spread up
to $2.5\times$ (Qwen3) / $7.6\times$ (DSv3) in true time, and
iso-activation sets up to $18\times$ / $24\times$---but the grid spans
token counts a single batch never mixes, so these bound the model's
curvature, not a dispatch decision.

The operational question is narrower: for one \emph{recorded} batch,
fixed expert set and fixed placement, how different are the dispatches
the actual proxy policies produce? \cref{tab:feasible} answers on the
recorded proxy traces (EP8, $R{=}64$, calibrated three-term model,
\cref{sec:serving-proxy}): the four proxy dispatches of the same batch
differ by $1.4$--$1.6\times$ in modeled block time at the median (p95
up to $1.7\times$), and the identity of the expensive proxy flips with
the regime---at $B{=}128$ (floor regime) the uniform replica split
pays $51\%$ over the modeled optimum while whole-expert activation
balancing (METRO) pays the least ($11\%$); at $B{=}512$--$1024$
(compute regime) static pays $47$--$56\%$ while token-LP \emph{is} the
optimum (the ensemble selects it). No fixed proxy is safe across the
ladder: the worst-regime median penalty ranges from $11\%$ (METRO) to
$56\%$ (static).

\begin{table}[t]
\centering\small
\caption{Modeled block time of each proxy dispatch relative to the
\sys{} solution, on recorded batches (median/p95 over windows and
layers; fresh placement, $R{=}64$). ``Spread'' = max/min across the
four proxy dispatches of the same batch.}
\label{tab:feasible}
\begin{tabular}{lccccc}
\toprule
$B$ & static & uniform & token-LP & METRO & spread \\
\midrule
128  & 1.17/1.27 & 1.51/1.60 & 1.17/1.27 & 1.11/1.18 & 1.37/1.46 \\
512  & 1.47/1.58 & 1.12/1.20 & 1.00/1.00 & 1.05/1.10 & 1.47/1.58 \\
1024 & 1.56/1.70 & 1.04/1.10 & 1.00/1.00 & 1.06/1.12 & 1.56/1.70 \\
\bottomrule
\end{tabular}
\end{table}

\subsection{Mixed regimes are the common case}
Using the calibrated $\nstar$ as classifier, \cref{fig:regime} asks
whether flat- and linear-region experts coexist within batches. Across
four workload traces (Qwen3-30B, DSv2-Lite $\times$ wikitext, GSM8K) at
128--512 tokens/GPU decode, 92--100\% of batches have ${\ge}20\%$ of
activated experts in the flat region \emph{and} ${\ge}20\%$ of tokens in
the linear region. At $B{=}512$, flat-region experts are 47\% of
activations while linear-region experts carry 91\% of tokens. Synthetic
Zipf routing ($s{=}1.2$) is mixed in 100\% of batches across the whole
batch ladder. The transition zone is not a corner case---it is the decode
operating range, where serving economics live.

\begin{figure}[t]
\centering
\includegraphics[width=\linewidth]{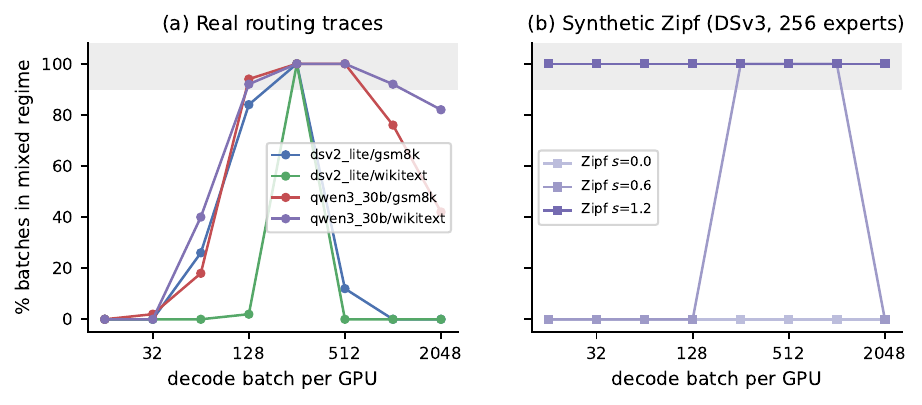}
\caption{Fraction of batches in the mixed regime (${\ge}20\%$ experts
flat and ${\ge}20\%$ tokens linear) vs.\ batch size. (a) Real traces;
(b) synthetic Zipf, DSv3 shape.}
\label{fig:regime}
\end{figure}

\section{Problem Formulation and the Phase Diagram}\label{sec:phase}

\subsection{Fixed-charge makespan dispatch}
Given per-expert token counts $n_e$ (from the router), replica sets
$R(e)$ (from the placement layer), token shares $x_{e,g}\ge 0$
($\sum_g x_{e,g}=n_e$, supported on $R(e)$) and activation indicators
$z_{e,g}=\mathbf{1}[x_{e,g}>0]$:
\begin{equation}
  \min\; \max_g\; \max\Bigl(a + b\textstyle\sum_e z_{e,g},\;
  c + \beta\textstyle\sum_e x_{e,g}\Bigr).
  \label{eq:problem}
\end{equation}

\begin{theorem}[NP-hardness]\label{thm:nphard}
Deciding whether a dispatch with makespan $\le T$ exists is NP-complete,
already with 2 GPUs, full replication, and $a=c=0$.
\end{theorem}
\begin{proof}[Proof sketch]
Reduce from Balanced PARTITION. Choose $b=T/(m/2)$, $\beta=T/W$; the max
form turns the single budget $T$ into two simultaneous per-GPU knapsack
constraints---an activation-cardinality cap $K=m/2$ and a token-capacity
cap $C=W$. Any split raises total activations above $2K$, so the
cardinality budget forbids splitting, and feasibility collapses to
balanced partition. Full proof in \cref{sec:proofs}.
\end{proof}

\begin{lemma}[$b\to 0$]\label{lem:lp}
With $a=b=0$ the objective reduces to $\min\max_g N_g$---exactly the
token LP (the LPLB/UltraEP objective); polynomial.
\end{lemma}
\begin{lemma}[$\beta\to 0$]\label{lem:semi}
With $c=\beta=0$ the objective reduces to $\min\max_g G_g$ over replica
sets---an optimal semi-matching~\cite{semimatching}; polynomial via
augmenting paths.
\end{lemma}

The problem is polynomial in each pure regime and NP-hard only when the
regimes interact---a formal mirror of the systems observation that each
proxy is fine at home and the transition zone is what requires a new
algorithm. This complements the recent NP-hardness of MoE serving at
GPU-quota granularity~\cite{moeserving}: ours is at per-batch dispatch
granularity under fixed placement, and neither implies the other.

Hardness notwithstanding, the problem admits a simple algorithm with an
\emph{additive} guarantee against the splitting-allowed optimum:

\begin{theorem}[Additive approximation]\label{thm:rr}
Under full replication, let $A_3$ place whole experts by descending-token
round-robin: GPU $i$ takes the experts ranked $i, i{+}g, i{+}2g, \dots$
Then $M(A_3)\le \mathrm{OPT}+\max(b,\,\beta n_{\max})$, where $n_{\max}$
is the largest per-expert token count and $\mathrm{OPT}$ may split
tokens arbitrarily; the additive term equals $\beta n_{\max}$ whenever
$n_{\max}\ge b/\beta$ (true in every recorded trace we use). The same
argument covers any per-GPU cost
$\max_k\bigl(a_k+\sum_e\varphi_k(n_e)\bigr)$ with nondecreasing
$\varphi_k\ge0$, with additive term $\max_k\varphi_k(n_{\max})$; for
the tile-aware model of \cref{eq:tilecost} this term is
$\max\bigl(b+b_2(\lceil n_{\max}/128\rceil{-}1),\,\beta
n_{\max}\bigr)$ (\cref{sec:proofs}).
\end{theorem}
\begin{proof}[Proof sketch]
Round-robin gives every GPU at most $\lceil E/g\rceil$ experts, matching
the activation lower bound, and row-wise domination on the sorted order
makes GPU~$g$ the lightest under \emph{every} monotone measure
simultaneously, so GPU~1 exceeds the mean by at most one expert. Full
proof in \cref{sec:proofs}.
\end{proof}

The additive term is unavoidable for whole-expert algorithms (a single
giant expert loses $\beta n_{\max}(1-1/g)$ against a perfect split).
Dividing by the lower bound
$\mathrm{LB}=\max\bigl(a+b\lceil E_{\mathrm{live}}/g\rceil,\;
c+\beta N/g\bigr)\le\mathrm{OPT}$ gives
$M(A_3)\le(1{+}\epsilon)\,\mathrm{OPT}$,
$\epsilon=\max(b,\beta n_{\max})/\mathrm{LB}$ ($=\beta
n_{\max}/\mathrm{LB}$ on all our traces): median $4.1\%$, max $7.7\%$
at recorded decode batches, growing to $31\%$ at $4\times$ token
scale---a coarse safety net, not the source of the practical gains.
Because the ensemble scores $A_3$ as a candidate, the deployed solver
carries the bound under full replication \emph{weakened by its
switching tolerance}:
$M(A^{*})\le(\mathrm{OPT}+\max(b,\beta n_{\max}))/(1-\tau)$ with
$\tau=1\%$ (\cref{prop:ensemble}); setting $\tau=0$ recovers the
additive form at the cost of the tolerance's near-tie protection.
With restricted replicas $A_3$ is generally infeasible and we have no
analogous guarantee (open problem, \cref{sec:limits}). The additive
term appears tight for $A_3$ (worst observed
$M(A_3)/(\mathrm{OPT}+\beta n_{\max})=0.998$ over 1296 instances
including adversarial giant-expert constructions; we do not have a
matching lower-bound proof), while \texttt{tempo\_fast} stays far
below it ($\le 1.14\times$ OPT worst case, $\le 1.0102\times$ on the
phase grid).

\subsection{The phase diagram}\label{sec:phasediag}
We sweep synthetic Zipf($s$) routing over batch size
$B\in\{16..2048\}$/GPU, skew $s\in\{0..1.5\}$, replication
$\{1.25,1.5\}\times$, EP $\{8..64\}$, with Dirichlet batch-to-batch
fluctuation ($\kappa{=}2000$), EPLB-style placement from long-run
averages, and the calibrated cost models (\cref{fig:phase}).

\begin{figure*}[t]
\centering
\includegraphics[width=0.84\textwidth]{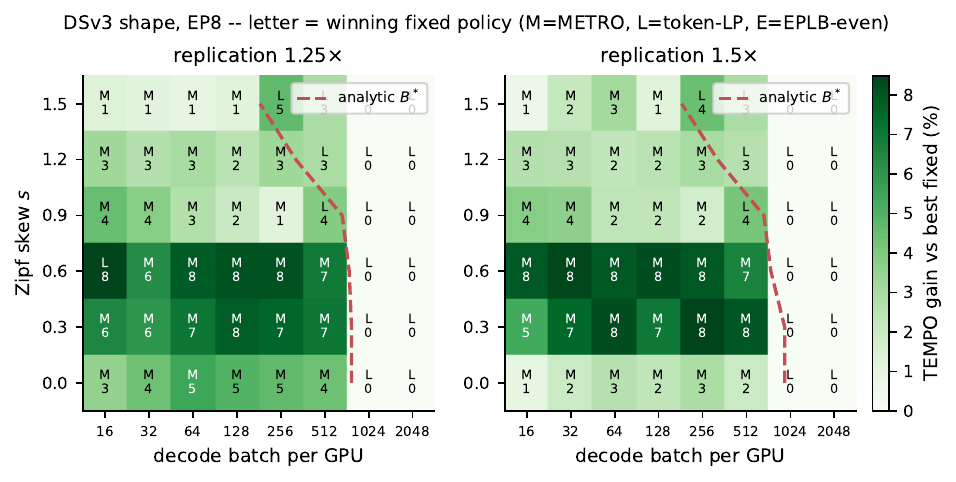}
\caption{Phase diagram, DSv3 shape, EP8. Letters mark the winning
\emph{fixed} policy (M = METRO-like activation balancing, L = token-LP,
E = EPLB-even); numbers and shading give \sys{}'s gain over the best
fixed policy per cell. The dashed line is the analytic flip boundary
$B^{*}$ of \cref{sec:phasediag}, with no per-panel fitting.}
\label{fig:phase}
\end{figure*}

\textbf{(i) The best fixed policy flips across the map}: activation
balancing wins at small $B$ (memory-bound), token-LP at large $B$
(compute-bound), and the boundary moves with $s$, replication, and shape.
There is no regime-free fixed proxy.

\textbf{(ii) The flip boundary is analytically predictable} by whichever
of two mechanisms triggers first. With
$E_{\mathrm{eff}}=\sum_e\bigl(1-(1-p_e)^{T}\bigr)r_e$ the expected number
of activated replicas ($T=BKn_{\mathrm{gpus}}$ total tokens, $r_e$
replica counts): (avg) the balanced cost pieces cross when
$\beta BK = b\,E_{\mathrm{eff}}/n_{\mathrm{gpus}}$, giving
$B^{*}_{\mathrm{avg}}=\nstar E_{\mathrm{eff}}/(K n_{\mathrm{gpus}})$;
(hot) activation balancing keeps experts whole, so the GPU holding the
hottest expert (popularity $p_1$) goes linear first,
$B^{*}_{\mathrm{hot}}=\nstar E_{\mathrm{eff}}/(p_1 K
n_{\mathrm{gpus}}^{2})$. $B^{*}=\min(B^{*}_{\mathrm{avg}},
B^{*}_{\mathrm{hot}})$, solved by a short fixed point in $B$, lands
inside the observed flip band in \textbf{12/12} (skew $\times$
replication) grid columns; the hot-expert mechanism takes over at
$s\ge 1.2$, which is why the boundary bends left under skew
(\cref{fig:phase}).

\textbf{(iii) \sys{} tracks the per-cell best everywhere} (min gain
$\ge -0.3\%$, inside the 1\% ensemble band) and wins by up to 8.5\%
(DSv3) / 10.2\% (Qwen3) / 11.8\% (DSv2, $1.5\times$ replication) in the
mixed zone.

\textbf{(iv) Mispricing is asymmetric}: METRO deep in the compute regime
costs up to $2.8\times$; token-LP in the memory regime costs
5--17\%---which explains why production noticed the METRO-direction
failure first.

\subsection{Scale extrapolation}
With the DSv3 model at EP 8$\to$64 (compute-only and comm-aware; comm
slope $k_r=0.0037\,\mu$s/token fitted from microbenchmark all-to-all
logs): the mixed-zone gain persists---mean 6--6.5\% at $s{=}0$ on
EP32--64 (max 15.5\%), 4.7--5.9\% at $s{=}1.2$ on EP32---and collapses to
0 exactly where theory predicts (the whole batch ladder enters the
compute regime; \sys{} coincides with token-LP). The communication term
changes no conclusion ($<$1\,pp) (\cref{fig:scale}).

\section{\sys{} Design}\label{sec:design}

\subsection{Black-box calibration on the deployed pipeline}
\sys{}'s cost table is fitted, not derived: run the deployed MoE pipeline
(or its microbenchmark) across a $(G,N)$ grid, log per-GPU $(G,N,t)$, fit
the two-piece max-affine by alternating assignment. About ten minutes per
(kernel, dtype, hardware) combination. \cref{sec:transfer} shows the loop
converges in one iteration (full-pipeline refit $\beta$: $0.358\to
0.356$) and that a \emph{mis}-calibrated table (GEMM-only parameters on a
pipeline with unfused per-token kernels) costs \sys{} 16\% at
$B{=}128$---calibration is load-bearing, not ceremonial. Every
parameter set used in the paper, and why the GEMM-only and
full-pipeline fits legitimately differ, is consolidated in
\cref{tab:allparams} (\cref{sec:calib-details}).

\subsection{The \texttt{tempo\_fast} solver}\label{sec:solver}
Four stages, each targeting a failure mode of naive approaches:
\begin{enumerate}
\item \textbf{Cost-aware greedy seeding.} Experts in descending token
  count, each placed whole on the replica with lowest marginal cost under
  \eqref{eq:cost}. Whole-expert placement matters: fragmenting cold
  experts is exactly the token-LP mistake.
\item \textbf{Activation rebalancing via augmenting chains.} In the flat
  regime the bottleneck is $\max_g G_g$; single moves stall because every
  direct move re-raises the destination. We use 1- and 2-step augmenting
  chains---the truncated version of the semi-matching augmenting-path
  algorithm that is optimal in the $\beta\to0$ limit (\cref{lem:semi}).
\item \textbf{Bottleneck local search with partial migrations.} In the
  linear regime, move token mass off the bottleneck GPU: full-expert
  moves and ternary-search partial splits (the objective restricted to
  one split size is piecewise-linear unimodal), accept the best
  bottleneck-reducing move, iterate ($\le 300$).
\item \textbf{Ensemble with a switching tolerance.} Token-LP and (when
  replication permits it) the round-robin placement $A_3$ of
  \cref{thm:rr} are feasible points of \eqref{eq:problem}; score each
  under the model and switch only if ${>}1\%$ better. Near-ties are
  decided by effects outside the $(G,N)$ model---measured hardware
  consistently prefers the heuristic's whole-expert structure
  (\cref{sec:limits})---so epsilon-level model differences must not
  trigger a switch. This makes ``never worse than token-LP (within
  1\%)'' hold by construction (\cref{prop:ensemble}); under full
  replication it also transfers the additive guarantee of \cref{thm:rr}
  to the deployed solver. The stage binds: without it, local search
  stops 2--4.5\% short of the LP optimum deep in the compute regime at
  EP32--64.
\end{enumerate}

\textbf{Optimality and cost.} Against a 10\,s HiGHS MILP on real-trace
grids: makespan ratio mean 1.005 / p95 1.024 / max 1.033; on the
synthetic phase grid $\le 1.0102$. Solve time: 1.9\,ms ($B\le 256$, pure
Python)---${\sim}550\times$ under the MILP---and ${\sim}20$\,ms for the
LP branch at large $B$, off the critical path.

\textbf{Component ablation.} Each stage binds in a different phase
region, none is redundant (\cref{tab:ablation}): removing partial
migrations degrades up to 16.3\% (EP32, $s{=}1.2$, $B{=}64$---transition
zone); removing augmenting chains up to 2.7\% (mid-$B$ flat zone);
removing the ensemble $\le 0.8\%$ at EP8 but up to 4.5\% at EP64
deep-compute.

\begin{table}[t]
\centering\small
\caption{Component ablation: makespan degradation vs.\ full
\texttt{tempo\_fast} over the phase grid (EP8) and scale grid
(EP8--64).}
\label{tab:ablation}
\begin{tabular}{lrrr}
\toprule
variant & mean & p95 & max (region) \\
\midrule
no partial moves & 0.6\% & 4.1\% & 16.3\% (EP32 transition) \\
no aug.\ chains  & 0.3\% & 1.8\% & 2.7\% (mid-$B$ flat) \\
no ensemble      & 0.4\% & 2.8\% & 4.5\% (EP64 compute) \\
\bottomrule
\end{tabular}
\end{table}

A \emph{cumulative} build-up on the drift-16 scenario
(\cref{tab:stageabl}) shows the two nontrivial stages are
regime-matched guards for each other. In the floor regime ($B{=}128$)
the seed alone recovers \emph{nothing} (1.000) and the augmenting
chains deliver the entire $-9.8\%$; in the traffic regime the roles
invert---chains \emph{alone} damage the seed's 0.883/0.869 to
0.940/0.932 (equalizing activation counts is wrong once traffic
binds) and local search repairs exactly that. Chains also warm-start
local search ($0.6$ vs.\ $2.2$\,ms at $B{=}128$); the ensemble is
worth $\le0.2$\,pp here, consistent with its role as an EP32+ safety
net.

\subsection{Communication- and topology-aware dispatch}
\label{sec:comm-topo}
Two extensions carry the model from the GEMM to the full deployed MoE
block. First, the a2a collectives scale with the tokens a rank receives,
so the deployed cost adds a third max-affine piece,
\begin{equation}
  t_g = \max\bigl(a + b\,G_g,\; c + \beta\,N_g,\; c_2 + \gamma\,N_g\bigr),
  \label{eq:cost3}
\end{equation}
with $(c_2,\gamma)$ fitted offline from recorded routing so the model
reproduces the measured static-vs-uniform crossover across batch sizes
(in simulation, $k_r$-style linear models from per-rank all-to-all logs
predict makespan within 12.7\%; all phase and scale conclusions are
unchanged, $<$1\,pp).

Does each term of \eqref{eq:cost3} earn its place? Re-solving the
drift-16 windows with terms deleted from the \emph{objective} but
scored under the full model (\cref{tab:termabl}): dropping the traffic
term forfeits more than half the win at $B{=}512$--$1024$
(0.946/0.924 vs.\ 0.880/0.865, p95 above static); dropping the floor
term is the mirror image (5\,pp worse at $B{=}128$). Dropping the
floor is \emph{identical} to bare token balancing---once the
activation term is gone, both remaining terms are monotone in
per-rank token count, so the two objectives share every argmin; the
floor is the only term that makes the objective non-token-reducible.
The terms are not monotone add-ons---floor \emph{without} traffic is
worse in the traffic regime than bare token balance, because the max
structure is what gates each term to its regime. A pure token-balance objective
looks safe at large $B$ under this (self-)evaluation, but its
$20$\,pp window-to-window swings end-to-end
(\cref{sec:serving-drift}) are what stage~4's switching tolerance
guards against.

Second, on multi-node EP the \emph{flat} traffic term in
\eqref{eq:cost3} treats every rank alike, but the replica a token is
sent to decides whether it crosses the inter-node network. We make the
dispatch topology-aware in two stages without touching the solver's
core: (1) solve \eqref{eq:problem} for per-GPU shares $x_{e,g}$ as
before; (2) for each expert, split the (node-uniform) token sources
across its replicas by a same-node-first transportation rule, which is
optimal for inter-node traffic among all splits that realize exactly the
solved $x_{e,g}$. The result is one dispatch table per \emph{source
node} rather than one global table; each rank uploads its own node's
slice, so the in-graph kernel is unchanged. Per-GPU loads---and hence
the compute makespan---are provably identical to the flat solution;
only the pairing of sources to replicas changes.
\cref{sec:serving-ep16} shows this distinction is worth
${\sim}11$\,pp of throughput at large batch on a 2-node deployment.

\subsection{SGLang integration}\label{sec:integration}
\sys{} plugs into SGLang's EPLB machinery as a dispatch algorithm. The
design goal is \emph{zero marginal in-graph cost}: the captured decode
graph must contain no extra kernels and no collectives.

\textbf{In-graph: one fused kernel, total.} The probabilistic dispatch
through the persistent logical-to-physical table (which LPLB also
needs) is fused with count collection: the sampling kernel atomically
increments a persistent \emph{cumulative} per-expert counter, masked
against CUDA-graph padding rows (whose degenerate routing would
fabricate a phantom hot expert---a bug we also found in the shipped
LPLB path). Cumulative counting removes the zeroing/copy kernels a
windowed counter would need; the background thread recovers windows by
wrap-safe differencing. Relative to static dispatch the marginal
in-graph cost is zero kernels; relative to LPLB it \emph{removes} the
in-graph interior-point solve and its per-layer EP collective.

\textbf{Off the critical path: a solver process, not a thread.} A
refresher thread snapshots the counters on a side stream (relaxed
stream-capture mode), aggregates window diffs across ranks on a
dedicated gloo group (shared groups are not thread-safe against the
scheduler's own collectives), and hands the solve to a separate
numpy-only worker process. In-thread versions cost 8--10\% end-to-end
from GIL contention (\cref{sec:serving-proxy}); out-of-process solving
eliminates it. Fresh tables are staged pinned and published with one
host-to-device copy.

\textbf{Consistency without double buffering.} Idle experts' rows are
untouched, invalid replica columns are zero in both tables, and the
kernel renormalizes each row at use, so a read racing the copy sees a
mixture of old and new rows---each a valid distribution. A single row
can also tear; in the worst tear it sums to zero, and for exactly this
case the kernel falls back to uniform over the replicas holding the
expert (one compare per row). A torn update costs one batch a stale or
uniform split; it can never route to a replica that lacks the expert,
nor divide by zero. The trade against in-graph solving is staleness
($\le$ 200\,ms default) for critical-path cost. Integration touches 7
patch points and survives EPLB rebalances (\cref{sec:patches}).

\section{Evaluation}\label{sec:eval}

Three evidence layers, from most controlled to most end-to-end: measured
cost grid $\to$ calibrated simulation (headline) $\to$ 8-GPU Testbed~A wall
clock $\to$ SGLang serving. Each layer states what it can and cannot
claim.

\subsection{Wall-clock microbenchmark (8-GPU Testbed~A, EP8)}
\label{sec:microbench}
Real pipeline per step: dispatch all-to-all $\to$ fp8 grouped GEMMs
$\to$ combine all-to-all, identical routed batches on all ranks,
CUDA-event timing, makespan via all-reduce MAX, medians over 30 steps
(\cref{fig:wallclock}); DSv3 expert shape, real Qwen3 routing trace.

At $B{=}32$ (memory-bound) \sys{} beats EPLB-even by 11--14\% and
token-LP by 7\%; token-LP is 19\% off best. At $B{=}2048$ (compute-bound)
METRO is 7--11\% off best; \sys{} ties token-LP at the top. Across all
$B$, \sys{} is within 5\% (${\approx}$ run noise) of the per-$B$ best
fixed policy while every fixed policy has a ${\ge}7\%$ failure region.
Both single-proxy failure directions are confirmed in wall clock, and the
``track the best everywhere'' property---the core claim---transfers.

\begin{figure}[t]
\centering
\includegraphics[width=\linewidth]{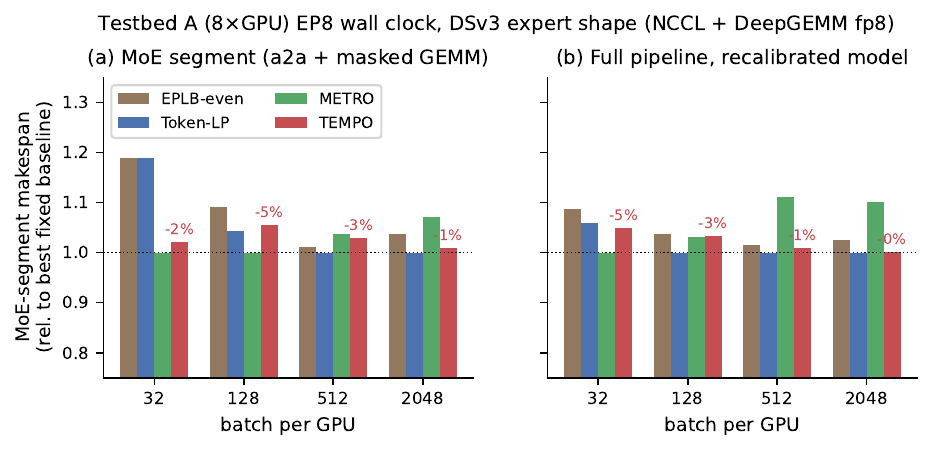}
\caption{8-GPU Testbed~A EP8 wall clock, DSv3 shape: per-batch makespan
relative to the best fixed policy of that batch size. (a) MoE segment;
(b) full pipeline with pipeline-recalibrated dispatch.}
\label{fig:wallclock}
\end{figure}

\subsection{Does the simulator transfer? Closing the calibration loop}
\label{sec:transfer}
We log per-rank $(G,N,t)$ inside the microbenchmark, refit \eqref{eq:cost}
black-box, and compare predicted vs.\ measured policy gains: DSv3 GEMM
pipeline---fit error 4.0\%, $\nstar{=}153$ (offline: 156), pairwise
ranking agreement 93\% (14/15 outside a 3\% noise band), mean
gain-transfer error 5.5\,pp; full pipeline with recalibrated
dispatch---92\%, 2.2\,pp, and the second-iteration refit reproduces the
first in every quantity that is identified (below): a one-iteration
fixed point.

\paragraph{What the deployed-pipeline fit identifies---and what it
does not.} The full-pipeline log only observes $G\in[9,20]$, so $a$
and $b$ are nearly collinear there: a cluster bootstrap (1000
replicates) puts $b\in[5.1,18.8]$ ($95\%$ CI;
$\mathrm{corr}(\hat a,\hat b)=-0.997$), which is why successive
refits report different $(a,b)$ splits (\cref{tab:allparams}:
$358/16.6$ vs.\ $488/5.9$)---one ridge, not a drifting model. What
the solver consumes is tightly identified: $\beta$ ($\pm0.5\%$) and
the flat cost $a+bG$ over the operating range ($\pm2.8\%$ at
$G{=}13$); re-solving the entire drift16 scenario under the two fits
moves a median $0.0\%$ of token mass and changes the chosen
dispatch's model makespan by $0.000\%$ (p95 $\le1.3\%$). The
\emph{mechanism} claim that $b$ is a weight-streaming floor rests on
the offline grid, which sweeps $G\in[1,40]$ and pins $b$ individually
(\cref{tab:params}).

Negative control: the Qwen3 shape \emph{cannot} anchor the simulator---its
$b\approx1.7\,\mu$s sits below measurement noise, the refit degenerates
(negative $b$), and identical assignment sequences vary $\pm$11\% across
runs from DeepGEMM tuner nondeterminism. We report this as a limitation
and note its flip side: where the model cannot discriminate, neither does
policy choice on hardware (confirmed in serving,
\cref{sec:serving}).

\subsection{Phase diagram and scale (calibrated simulation)}
\cref{sec:phasediag} numbers; \cref{fig:phase,fig:scale}. No fixed policy
is safe across the map; \sys{} min gain $\ge-0.3\%$ / max 8.5--15.5\%;
gains persist to EP64; comm-aware unchanged.

\subsection{Real-trace replay}
Qwen3-30B and DSv2-Lite traces $\times$ wikitext/GSM8K under the
calibrated models: \sys{}'s gap to a per-cell oracle (a fixed-policy
oracle switching per cell---unrealizable in deployment) is 6--16\%; with
DSv3 costs 11--14\%. \texttt{tempo\_fast} matches the MILP within 0.8\%
throughout.

\subsection{Pricing staleness: the deployed system vs.\ the theory}
\label{sec:staleness}
The theory optimizes the \emph{current} batch; the deployment solves
on the \emph{previous} window's counts. Replaying the recorded proxy
routing under four information regimes (evaluated on true current
counts, calibrated model, $R{=}64$, decode scale): \emph{exact}
reaches $0.78\times$ static (median $0.80$), \emph{stale-1} (the
deployment) $0.85$, \emph{frozen} (never refreshed) $0.89$,
\emph{uniform} (no solve) $1.09$. One window of staleness returns
$5$--$7$\,pp of the $22$\,pp exact-solve headroom and keeps the rest,
never flipping the sign in the mean; the pattern is stable across
$R\in\{16,64\}$ and $1$--$4\times$ token scale (stale-1 retains
$75$--$85\%$ of the exact gain, frozen $70$--$80\%$)---consistent
with the serving result that the refresh interval is a $\sim$2\,pp
knob (\cref{sec:serving235b}). The regime that actually inverts is
uniform splitting, whose end-to-end signature is the $-5.4\%$
tight-budget regression of \cref{sec:serving-drift}.

\subsection{Dispatch vs.\ refreshing the placement faster}
\label{sec:refresh-control}
If placement repairs the mean, why not refresh EPLB more often and
skip the dispatcher? Replaying the drift16 timeline with the placement
rebuilt every $W$ windows from observed counts (\cref{tab:refresh}):
the answer is regime-dependent and composable, not either/or. In the
compute regime ($B{=}512$), production EPLB (uniform split +
cumulative refresh) reaches $0.85\times$ static-stale; \sys{} on the
\emph{never-refreshed} placement reaches $0.89$ without moving a
single weight; the composition is best ($0.80$) with the best tail.
In the floor regime ($B{=}128$) fresher placement makes the uniform
split \emph{worse} ($1.14\to1.31$)---a more balanced replica spread
fragments more experts and manufactures more activation floors---while
\sys{} holds parity; short trailing-window refresh chases noise (p95
$1.38$). Not charged: an EPLB refresh copies GB of weights, a table
refresh copies $256\times G$ floats.

\subsection{Token-moving vs.\ weight-moving (MoonEP-style)}
\label{sec:moonep}
Weight-moving balancers (MoonEP~\cite{moonep}) duplicate experts online
for perfect token balance. Replaying that policy on recorded
Qwen3-235B routing and racing MoonEP itself against DeepEP v2 on
8-GPU Testbed~B (\cref{sec:moonep-details}) places it on the phase
diagram rather than in competition: at decode scale its weight
prefetch alone ($283$--$502\,\mu$s) exceeds the whole MoE-layer
budget ($95$--$107\,\mu$s for DeepEP dispatch+combine), because in the
memory-bound regime moving an expert's weights costs about as much as
using them; at prefill scale with extreme skew (MaxVio
${\gtrsim}20$) it wins, owning the compute-bound high-skew corner.
Token-moving is the only lever whose marginal cost stays zero at
decode.

\subsection{End-to-end serving (SGLang, DeepSeek-V2-Lite, EP8)}
\label{sec:serving}
Deployability evidence, reported honestly: on a small-expert model
($b\approx3\,\mu$s) with bf16/Triton and DeepEP-fp8 configurations,
adaptive dispatch (\sys{} and LPLB alike) does not beat static
placement---consistent with the cost model's own prediction (small $b$
$\Rightarrow$ small dispatch stakes) and with the negative control of
\cref{sec:transfer}. Within the adaptive class, \sys{} $\ge$ LPLB on 2/3
workloads (median TPOT $-7$ to $-9\%$). The integration itself is the
claim: CUDA-graph-safe, zero-overhead in-graph dispatch, stable under
EPLB rebalances, correct generations.

\textbf{Refresher-period sensitivity.} Sweeping $\{10,50,200\}$\,ms:
10 and 50\,ms are indistinguishable while 200\,ms is ${\sim}25\%$
faster with half the TTFT---not staleness (popularity drifts far
slower than 200\,ms) but the in-thread solver's GIL footprint, which
the out-of-process solver of \cref{sec:integration} later removes
(\cref{sec:serving-proxy}). We default to 200\,ms.

\subsection{End-to-end serving at scale (SGLang, Qwen3-235B-FP8, EP8)}
\label{sec:serving235b}
We repeat the serving study on Qwen3-235B-A22B-FP8 (8-GPU Testbed~A, EP8,
DeepEP low-latency, 94 MoE layers, 128 experts top-8 with 32 redundant
physical slots, placement initialized from a recorded expert
distribution of the same traffic). Two traffic classes: synthetic
decode-heavy (random prompts; near-uniform routing) and ShareGPT (real
text; per-layer max/mean expert load $2.7$--$4.1\times$).

\textbf{\sys{} vs.\ LPLB (like-for-like machinery).} Both use the
identical probabilistic-dispatch path, isolating \emph{what is solved
and where}: \sys{} delivers $+38$--$70\%$ request throughput on decode
workloads and $+65\%$ on ShareGPT at roughly half the median
inter-token latency. The mechanism is architectural---LPLB's in-graph
solve embeds a per-layer EP collective that a 94-layer model replays
every step, while \sys{}'s zero-collective path pays a local bincount.
(Running LPLB at all required relaxing its model whitelist; upstream
disables it outside DeepSeek architectures because the in-graph
collective deadlocks under DP-attention empty batches, a failure mode
background solving avoids by construction.)

\textbf{Against static placement, honestly.} With the v1 integration,
well-initialized static led \sys{} by $9$--$20\%$: the tax of the
remap machinery, paid by the whole adaptive class. On the v2 stack
(fused kernel, out-of-process solver) the gap shrinks to $2.3\%$ on
ShareGPT and $3$--$11\%$ on decode, and \sys{} beats its own noop
control on all six configurations ($+0.2$ to $+1.6\%$): the solver is
free, and the residual gap is the remap data path a 94-layer model
multiplies. With fresh placement and near-uniform routing there is
nothing left for dispatch quality to win back; on the same v2 windows
LPLB trails static by $43$--$48\%$. Two shipped-code hazards surfaced
and are fixed in our patch: solver input taken before padding-masking
fabricates a phantom hot expert (LPLB's shipped path has the same
defect), and background-thread collectives need a dedicated process
group (\cref{sec:integration}).

\subsection{Cross-model serving on Testbed~B: the win region, bracketed
(fp8 DeepGEMM, EP8)}\label{sec:serving-tb}
Our primary at-scale evidence moves to Testbed~B: two flagship
models, native fp8 DeepGEMM masked experts, tile-aware parameters
\eqref{eq:tilecost}, single-node EP8, all policies back-to-back per
window. The models bracket the phase diagram's win region:
Qwen3-235B-FP8 (94 MoE layers, 16 logical experts per GPU) sits
inside; DeepSeek-V3-0324 (61 layers, 32 per GPU) sits outside, where
statistical averaging over $4\times$ more experts per GPU flattens
the imbalance dispatch could repair. Traffic: synthetic decode-heavy
random, ShareGPT, four real datasets as-is (OASST1, GSM8K, a bilingual
Q\&A corpus, GovReport long-document summarization), plus a Poisson
arrival sweep (\cref{fig:tbserving}).

\textbf{Qwen3-235B (inside).} \sys{} wins where the phase diagram says
it should, and the wins replicate across three dedicated repeat
windows plus two earlier ones (\cref{fig:robust}a,b): GovReport
(long-prefill, compute-bound experts) throughput $+5.0\%$ median with
non-overlapping ranges (521--524 vs.\ 549--554 tok/s), and under
moderate Poisson load (8\,req/s) p99 TPOT $-15.6\%$ (191 vs.\
226\,ms) with median TTFT $-12.5\%$---tail latency is where transient
routing skew concentrates and where a per-batch time model pays.
ShareGPT is parity ($-0.2\%$ throughput, $-8.9\%$ p99); saturated
decode-heavy random trails static by the residual data-path tax
($2$--$4\%$). A refresh-period sweep at $\{200,500,1000\}$\,ms leaves
that deficit unchanged---so it is neither solver pressure nor
staleness---while the frozen-table control is \emph{worse} ($-5.0\%$
vs.\ $-2.2\%$): live updates recover about half of a constant
dispatch-path tax. The noop control sitting between static and \sys{}
splits the credit: roughly two thirds of the GovReport gain is the
uniform-split data path, the last third the solver, while the tail win
is mostly the solver's.

\textbf{DSv3 (outside).} Every workload lands at $-2$ to $-3\%$,
indistinguishable from the noop control. The clean diagnostic is
SGLang's own dynamic rebalancer: it also gains nothing ($\pm1\%$)---at
32 experts per GPU there is no exploitable imbalance left, so
\emph{any} adaptive machinery can only pay its own cost. The pair of
models turns the phase diagram from a simulation artifact into a
falsifiable, and here twice-confirmed, deployment rule.

\textbf{LP, everywhere.} The token-LP dispatcher collapses on both
models and all workloads ($-10\%$ to $-56\%$): the per-layer in-graph
collective compounds with depth (61 and 94 layers), and token balancing
fragments activations exactly as \cref{sec:phase} predicts.

\begin{figure*}[t]
\centering
\includegraphics[width=\textwidth]{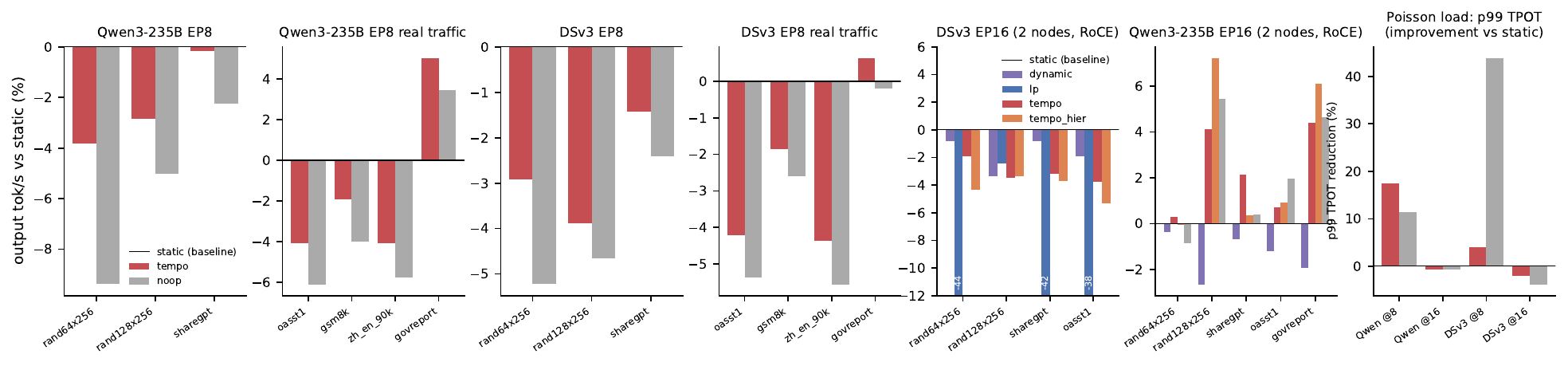}
\caption{Testbed~B serving, all policies normalized to same-window static
(the $y{=}0$ baseline). Left: closed-loop throughput on Qwen3-235B
(inside the win region) and DSv3 (outside). EP16 panels: DSv3 is
communication-dominated and no balancer helps; on Qwen3-235B the
\sys{} family wins $+2$--$7\%$ with the topology-aware split ahead.
Right: p99 TPOT under Poisson arrivals---\sys{}'s largest win (Qwen
$-15.6\%$ at 8\,req/s) and the noise floor of tail metrics.}
\label{fig:tbserving}
\end{figure*}

\subsection{Eliminating the adaptive-dispatch tax
(DSv3-shape proxy, EP8)}\label{sec:serving-proxy}
The remaining serving studies isolate the flagship-shape regime our
testbed cannot host natively: a reduced-depth proxy with DSv3's exact
expert geometry (hidden 7168, 256 routed experts, top-8; 1 dense + 4
MoE layers; fp8; Zipf-skewed router), SGLang on 8-GPU Testbed~A, DeepEP
low-latency, 64 redundant experts ($1.25\times$), EPLB placement from
recorded routing. All comparisons run back-to-back in the same idle
window (cross-window drift on shared machines reaches 15\%, larger
than any effect studied).

An instrumented \emph{noop} variant---full \sys{} data path, solver
never runs, table frozen at uniform-over-replicas---decomposes overhead
cleanly. With the v1 integration (in-thread solver, separate count
kernels), \sys{} trailed static dispatch by 8--14\%: solver GIL
contention cost 8--10\% (refresh $200\to1000$\,ms recovered most of it)
and the extra in-graph kernels cost 7--12\% at small batch. The three
fixes of \cref{sec:integration,sec:comm-topo}---out-of-process solver,
fused dispatch+count kernel, three-term cost---eliminate the tax
entirely: \sys{}, noop, static, and SGLang's dynamic dispatcher land
within $\pm 2$--3\% (run noise) at all three load levels, and \sys{} has
the lowest TPOT at small batch. Adaptive dispatch is now free even where
it cannot win: with fresh placement and $1.25\times$ replication, EPLB
has already absorbed the imbalance, and the honest headline is parity,
not gains.

\subsection{When dispatch pays: stale placement and tight replica
budgets}\label{sec:serving-drift}
Placement is recomputed on a minutes-scale horizon; between rebalances,
routing drifts. We emulate this by building the EPLB placement and
static maps from a corrupted recording (the per-layer top-$K$ hot
experts' statistics swapped with random cold ones, $K{=}16$) while
serving the \emph{true} traffic---then repeat every comparison in three
independent windows (\cref{fig:drift}).

Three findings. \textbf{(i) Architecture and objective, separated.}
Against SGLang's shipped LP dispatcher (in-graph per-layer solve),
\sys{} is $+12.5\%$ at 128 concurrent requests and $+6.6\%$ at
1024---but that conflates architecture with objective. Porting the
exact token-LP into our out-of-process worker (same fused kernel, same
cadence, only the objective differs) and rerunning drift-16 in three
fresh paired windows: token-LP matches \sys{} at low and mid load
(both $+9.5\%$ at 128; parity at 512), and the residual difference
appears exactly where the model says fragmentation hurts---at 1024
requests token-LP swings $-11.6\%$ to $+8.7\%$ across windows while
\sys{} stays within noise in every window (spread 4.4 vs.\ 20\,pp).
Most of the headline gap was architecture; the objective's own edge is
\emph{stability} at the compute-bound point. \textbf{(ii) Adaptive
beats stale-static at low and mid load} ($+3\%$, 3/3 windows), mostly
captured by the uniform-split data path; at 1024 all policies tie (an
earlier single-window $+5\%$ did not replicate; we retract it).
\textbf{(iii) The solver earns its keep when replicas are scarce.} At
a $1.06\times$ budget, blind uniform splitting \emph{backfires}
($-5.4\%$: probability mass flows to mistakenly-replicated cold
experts) while the cost model recovers parity ($+0.3\%$). A model of
\emph{where} tokens should go is what makes replication safe to
exploit.

\paragraph{Real model.} The same corruption applied to Qwen3-235B
itself (Testbed~B, fp8, EP8; two windows agreeing within $0.5\%$;
\cref{fig:robust}c): stale placement costs static dispatch $-7.3\%$
ShareGPT throughput and $+40\%$ p99 TPOT. \sys{} claws back a third of
the tail damage ($-10\%$ p99, $-5\%$ TTFT) but none of the
throughput---and SGLang's dynamic rebalancer recovers neither. The
split is structural: under drift the redundant replicas sit on the
wrong (cold) experts, so the throughput loss is a \emph{provisioning}
error only placement can fix, while the tail comes from transient
co-scheduling of hot experts that dispatch can still steer around.
Dispatch repairs the tail; only placement repairs the mean.

\subsection{Locating the win boundary: $B^{*}_{\mathrm{a2a}}$
predicted, then measured}\label{sec:bstar}
The three-axis rule so far names the win conditions; here we turn it
into a number and test it. This crossover is distinct from the flip
boundary $B^{*}$ of \cref{sec:phasediag} (floor vs.\ compute, which
policy is best): here we ask where adaptive dispatch stops paying at
all, which under the calibrated model happens where the a2a traffic
term overtakes the activation floor: $N^{*} = (a + b\,G -
c_2)/\gamma$ tokens per GPU. With the deployed proxy parameters
(\cref{tab:allparams}: $a{=}16.4$, $b{=}2.99$, $c_2{=}25.0$,
$\gamma{=}0.10$) at the solver's observed activation count
($G \approx 13$ per GPU), $N^{*} = (16.4 + 38.9 - 25.0)/0.10 \approx
303$, i.e.\ $B^{*}_{\mathrm{a2a}} \approx 300$---the simulated
dominant term flips between $B{=}256$ and $384$. We then swept
concurrency 32--1024 end-to-end (static vs.\ \sys{}, drift-16,
$R{=}64$, two paired windows per point); at EP8 decode with $K{=}8$
experts per token, a step at concurrency $C$ routes $CK/8\approx C$
token--expert pairs to each GPU, so serving concurrency is
numerically comparable to the tokens/GPU axis $B$ of
\cref{sec:phasediag}, and we reuse the symbol. We separately
measured, with
a torch-profiler kernel breakdown of steady-state decode, the expert
block's share of GPU busy time (DeepEP dispatch/combine + grouped
GEMM + gating; the Amdahl denominator).

\begin{table}[t]
\centering\small
\caption{$B^{*}_{\mathrm{a2a}}$ prediction vs.\ measurement (drift-16, $R{=}64$):
modeled regime and block ratio vs.\ paired-window throughput delta
and expert-block share of GPU busy time.}
\label{tab:bstar}
\begin{tabular}{rlrrr}
\toprule
$B$ & regime & block ratio & $\Delta$thr.\ (w1/w2) & expert share \\
\midrule
32   & floor   & 0.902 & $-1.5\%$/$-1.5\%$ & 46\% \\
64   & floor   & 0.902 & $+4.6\%$/$+4.1\%$ & --- \\
128  & floor   & 0.902 & $+3.3\%$/$+2.9\%$ & 51\% \\
256  & floor   & 0.907 & $+4.0\%$/$+2.6\%$ & --- \\
512  & traffic & 0.880 & $-2.8\%$/$+0.0\%$ & 31\% \\
768  & traffic & 0.871 & $-1.3\%$/$-2.1\%$ & --- \\
1024 & traffic & 0.865 & $-0.6\%$/$-1.1\%$ & 27\% \\
\bottomrule
\end{tabular}
\end{table}

Three reads of \cref{tab:bstar}. \textbf{The boundary lands where
predicted}: gains are positive in both windows at $B{=}64$--$256$ and
vanish from $B{=}512$ on, bracketing
$B^{*}_{\mathrm{a2a}}\!\approx\!300$.
\textbf{Inside the band, Amdahl accounting closes}: block-level model
gain ($9.8\%$) times measured expert share ($0.46$--$0.51$) predicts
$4.5$--$5.0\%$ end-to-end; measured $+2.6$--$+4.6\%$. \textbf{Outside
the band the model is optimistic}: past $B^{*}_{\mathrm{a2a}}$ the expert share
shrinks to $27$--$31\%$, and the modeled traffic-regime gains do not
materialize at all---once the a2a is bandwidth-bound, uniform
splitting already balances traffic and the remaining modeled headroom
is not realizable by dispatch; at $B{=}32$ the modeled differences sit
below the launch-and-sync floor the $(G,N)$ model deliberately omits.
Operators get a band: adaptive dispatch pays above the batch where
per-expert traffic amortizes the activation floor and below
$B^{*}_{\mathrm{a2a}}$;
the shipped calibration predicts the upper edge, the lower edge is
found empirically.

\subsection{Policy landscape: inference- and training-side kernels under
one cost model}\label{sec:landscape}
Training-side balancers (LLEP~\cite{llep}, TAOT~\cite{taot},
UltraEP~\cite{ultraep}) cannot be reproduced end-to-end on a serving
stack---they exploit weight-migration windows between gradient steps
that serving does not have. What \emph{can} be compared
like-for-like is each scheme's dispatch-policy kernel: we re-implement
five of them under the identical replica constraint (move tokens, never
weights) and evaluate block makespan on the recorded proxy routing under
the calibrated three-term model, sweeping the redundancy~$R$ $\times$
drift~$K$ grid of \cref{sec:serving-drift}
(\cref{tab:landscape}). \emph{Token-LP} is the exact HiGHS min-max
token count---simultaneously LPLB's objective and UltraEP's
``exact post-reroute load'' quota at dispatch granularity;
\emph{METRO} assigns whole experts to balance activation counts;
\emph{LLEP-R} is LLEP's least-loaded greedy spill restricted to existing
replicas.

\begin{table}[t]
\centering
\caption{Block makespan relative to static (lower is better), recorded
proxy routing, calibrated 3-term cost, $R{=}64$. Columns are batch
scales spanning the three phases. Boldface: best per column.}
\label{tab:landscape}
\small
\begin{tabular}{lccc}
\toprule
Policy kernel & \makecell{$0.25\times$\\(floor-bound)} &
\makecell{$1\times$\\(mixed)} & \makecell{$4\times$\\(token-bound)} \\
\midrule
uniform (dynamic)        & 1.299 & 1.052 & 0.659 \\
token-LP (LPLB/UltraEP)  & 1.006 & 0.836 & \textbf{0.630} \\
METRO (activation)       & 0.950 & 0.815 & 0.672 \\
LLEP-R (least-loaded)    & 0.984 & 0.835 & 0.662 \\
\sys{} (time model)      & \textbf{0.856} & \textbf{0.764} &
\textbf{0.630} \\
\bottomrule
\end{tabular}
\end{table}

Scoring every policy under the model \sys{} optimizes is
self-evaluation, so we do not rest the claim on it alone: the model is
validated against wall clock (\cref{sec:transfer}), the token-LP row
has its like-for-like anchor (\cref{sec:serving-drift}), and we ran
METRO and LLEP-R end-to-end as alternative table builders in three
drift-16 windows each. Median throughput vs.\ same-window static at
512/1024 requests: METRO $-2.1\%$/$-0.2\%$, LLEP-R $+0.1\%$/$-2.2\%$,
\sys{} $-0.3\%$/$+2.0\%$; at 128 requests window variance ($\pm8\%$)
swamps all policies. The wall-clock ordering matches the table's
prediction, compressed by Amdahl's law, and no proxy kernel beats
\sys{} at any load point.

The table is the phase diagram with named competitors. In the
floor-bound phase token-LP \emph{loses to static} (splitting tokens
pays activation floors) and METRO is the best proxy---still 10\,pp
behind \sys{}, the only policy that prices the floor. In the
token-bound phase \sys{} coincides with token-LP exactly (the ensemble
selects it), so the time model costs nothing where proxies are already
right; in the mixed phase it leads the runner-up by 5--7\,pp. At
$R{=}16$, $K{\ge}16$ \emph{every} kernel collapses to $1.000$: when
the drifted hot set has no replicas, the fix belongs to the placement
layer (LLEP's full mechanism, TAOT's optimal transport, an EPLB
recompute). Dispatch and placement partition the problem; \sys{} is
complementary to all of them.

\subsection{Multi-node EP16 on Testbed~B: topology decides}
\label{sec:serving-ep16}
We repeat the proxy campaign on $2\times 8$ Testbed~B (CUDA 13.2,
RoCE $8\times400$G per node), TP16/EP16, DeepEP internode low-latency;
bring-up hazards worth two orders of magnitude of performance
(NCCL plugin and transport pathologies, fp8 scale packing on the new SM architecture)
are documented in \cref{sec:bringup}.

\emph{Window rule (pre-declared).} A paired window benches a policy
and static back-to-back on idle, dedicated nodes; a window is
excluded only for an infrastructure failure identifiable without
reference to which policy won (crash, reclaimed node, or physically
impossible bench-client metrics)---never for its result. All tail
metrics are reported, including those against us.

With fresh placement, static, \sys{}, and noop tie within noise across
three windows---the zero-tax result, including the 16-rank count
aggregation, carries to multi-node unchanged. Under the drift-16
scenario we accumulated five paired windows per policy (two further
windows hit the bench-client artifact at 1024 requests---30\%
throughput jumps with a physically impossible 0.03\,ms inter-token
latency---and are excluded under the rule above). Medians against
same-window static: at the all-to-all-bound operating point (1024
requests) the flat table \emph{loses} $-3.5\%$---spraying tokens
uniformly across a hot expert's replicas sends half of them over the
network, and the flat $\gamma N$ term cannot see the difference---while
the two-stage topology-aware split of \cref{sec:comm-topo} recovers
$+4.1\%$ and the token-LP dispatcher, whose per-GPU token balancing
incidentally balances receive traffic, $+4.0\%$: an ${\sim}8$\,pp swing
attributable purely to \emph{which source feeds which replica}
(\cref{fig:hier}). At 512 requests all adaptive policies sit
$2$--$5\%$ below static; at 128, all within noise. LP being viable here
and catastrophic on Qwen3-235B ($-43$\% to $-48$\%,
\cref{sec:serving235b}) is the depth story in one system: its in-graph
per-layer solve is amortizable over 4 MoE layers and fatal over 94
(running LP on Testbed~B at all required a cublasdx fix we upstream,
\cref{sec:bringup}).

Sweeping the staleness level completes the picture: at drift-8 the
placement is still mostly right and the hier split pays a consistent
$3$--$6\%$ at mid/high load with nothing to win back; at drift-32 the
hot set has rotated so far that no node holds local replicas for
it---repair now requires moving \emph{weights} (placement's job), not
tokens, and hier sits $4$--$7\%$ under static. Dispatch-layer
adaptivity on multi-node EP pays inside a staleness band: enough drift
that placement is wrong, not so much that only placement can fix it.
On multi-node EP, locality is not an optimization on top of load
balance; it is part of the cost model or it is a regression.

\paragraph{Real model, real weights: communication can evict the whole
question.} Repeating the window with DeepSeek-V3-0324 itself (fp8, 61
MoE layers, 320 physical experts, fresh placement) gives the phase
diagram its third axis. Logical experts per GPU drop from 32 (EP8) to
16---inside the win region by the two-axis rule---yet every balancer
lands at or below static. Across three windows spanning two node
pairs, medians against same-window static: dynamic $-0.8$ to
$-3.4\%$, \sys{} $-1.8$ to $-3.8\%$, the topology-aware split $-3.4$
to $-5.3\%$, LP $-38$ to $-50\%$ on decode-heavy workloads yet only
$-2.5\%$ on the prefill-heavy one (its per-layer solve amortizes over
long steps: the depth story within a single model). The diagnostic is
absolute: 16 GPUs over RoCE deliver \emph{less} than 8 over NVLink
(9.8 vs.\ 12.7\,req/s), because inter-node dispatch/combine and
allreduce dominate the step; the windows replicate to the percentage
point (LP's oasst1 leg: $-49.5$ vs.\ $-49.2\%$). Tails carry the same
message, reported in full: \sys{}'s p99 TPOT is stable
($869$--$917$\,ms) while static's varies ($594$--$896$\,ms), so the
per-window gap ranges $+1.4\%$ to $+47.9\%$ against us depending on
which static tail one draws (median-of-windows $+2.5\%$). We find no
view of this data in which adaptive dispatch helps here: when the
expert-FFN stage is a minor fraction of step time, balancing it moves
nothing and every policy returns its mechanism cost. The win region
needs three coordinates: moderate experts-per-GPU, sufficient skew,
\emph{and} expert compute a large enough share of the step.

\paragraph{The third coordinate is compute share, not node count.}
Qwen3-235B on the same two nodes pins the axis down: 8 logical
experts per GPU and a smaller expert geometry keep the FFN stage a
substantial share of the step even over RoCE---and here multi-node
dispatch \emph{wins} (\cref{fig:tbserving}, Qwen EP16 panel).
Medians against same-window static: decode-heavy \sys{} $+4.1\%$
output throughput and topology-aware split $+7.2\%$ (three windows
each, all positive); GovReport $+4.4$/$+6.1\%$; chat
$+0.4$--$2\%$; SGLang's dynamic rebalancer negative throughout. The
four-way decomposition (\cref{tab:decomp}) attributes the gain: the
no-solve control captures part ($+5.4\%$ on decode-heavy), the flat
solve matches it, and the same-node-first split adds the
rest---consistent with a fabric where \emph{which source feeds which
replica} carries an ${\sim}8$\,pp swing (\cref{fig:hier}). The table
also reports the legs against us: on Poisson chat all adaptive
variants sit $0.6$--$2.9\%$ below static; short decode is noise.

\begin{table}[t]
\centering\small
\caption{Four-way attribution on Qwen3-235B 2-node EP16 (Testbed~B): median
output-throughput gain vs.\ same-window static over three paired
windows. \emph{dynamic} = SGLang's rebalancer; \emph{noop} = \sys{}
data path without solving (uniform replica split); \emph{flat} =
\sys{} solve, one shared table; \emph{hier} = \sys{} solve + same-node
-first source split.}
\label{tab:decomp}
\begin{tabular}{lrrrr}
\toprule
workload & dynamic & noop & flat & hier \\
\midrule
decode-heavy (128/256) & $-2.6\%$ & $+5.4\%$ & $+4.1\%$ & $+7.2\%$ \\
GovReport (long prefill) & $-1.9\%$ & $+4.6\%$ & $+4.4\%$ & $+6.1\%$ \\
ShareGPT & $-0.7\%$ & $+0.4\%$ & $+2.1\%$ & $+0.4\%$ \\
oasst1 (saturated) & $-1.2\%$ & $+2.0\%$ & $+0.7\%$ & $+0.9\%$ \\
oasst1 (Poisson 16\,req/s) & $-0.6\%$ & $-2.5\%$ & $-1.8\%$ & $-2.9\%$ \\
short decode (64/256) & $-0.4\%$ & $-0.8\%$ & $+0.3\%$ & $-0.1\%$ \\
\bottomrule
\end{tabular}
\end{table}

Is the hier column pure locality, the solve along for the ride? A
dedicated control under drift-16 placement (\cref{tab:noslv}) runs a
\emph{uniform} no-solve policy through the identical worker path, with
and without the same-node-first split. No-solve dispatch loses to
static in all six paired windows regardless of locality (hier-uniform
$-1.4$ to $-7.0\%$); the full solve beats the no-solve control in six
of six ($+0.3$ to $+4.0$\,pp); locality alone contributes only at the
all-to-all-bound point ($+1.0$\,pp at 1024). The same-node-first split
is a multiplier on a good table, not a substitute for one.

Together the two models give
the three-axis rule its cleanest statement: DSv3 EP16 fails the
compute-share test on the same fabric where Qwen EP16 passes it, so
what evicts dispatch at multi-node scale is not the second node but
the step-time share that balancing can still touch.

\section{Related Work}\label{sec:related}

\textbf{Token-proxy balancing.} EPLB~\cite{eplb} (placement, average
load); LPLB~\cite{lplb} (per-batch token LP on replica graphs; its
documentation flags nonlinear expert cost as unsolved---the problem
statement we complete); FlexMoE~\cite{flexmoe} and
SmartMoE~\cite{smartmoe} (training-time). All assume time $\propto$
tokens; FineMoE~\cite{finemoe} targets the orthogonal memory axis.

\textbf{Activation-proxy balancing.} METRO~\cite{metro} balances
activated experts for memory-bound decode; expert-level replication
tuning is orthogonal and composable with our EPLB-style placement.

\textbf{Weight-moving balancing.} MoonEP~\cite{moonep} reaches perfect
per-rank token counts by prefetching redundant experts online---moving
\emph{weights} where we move \emph{tokens}; \cref{sec:moonep} shows
the policy is activation-blind at decode and bandwidth-negative once
weight movement is charged---the two mechanisms partition the phase
diagram. UltraEP~\cite{ultraep} reroutes per microbatch toward exact
token counts (its quota objective coincides with token-LP at dispatch
granularity, so \cref{tab:landscape} brackets it); TAOT~\cite{taot}
places training-time guest replicas by optimal transport over a
topology cost matrix---the placement-layer dual of our topology-aware
split, and the mechanism our $R{=}16$/$K{\ge}16$ collapse row calls
for.

\textbf{Latency-aware, different axis.} ViBE~\cite{vibe} balances time
across heterogeneous \emph{hardware} with per-device linear rates; the
regime nonlinearity we model is exactly what it misses.
LLEP~\cite{llep} packs greedily by estimated latency without an
optimality framework or validated cost model; CAEE~\cite{caee} drops
computation (lossy); kernel/overlap optimizations
(e.g.~\cite{deepep,deepgemm}) change $(a,b,\beta)$, which
recalibration absorbs.

\textbf{Theory.} MoE-Serving~\cite{moeserving} proves NP-hardness at
GPU-quota granularity; \cref{thm:nphard} complements it at dispatch
granularity, and \cref{thm:rr} adds the positive side.
Semi-matching~\cite{semimatching} and scheduling with job splitting
and setup times~\cite{xingzhang} supply the degenerate-case
algorithms. To our knowledge \sys{} is the first EP dispatcher that
models both regimes in one calibrated cost function, optimizes the
resulting makespan with per-batch guarantees against both classical
proxies, and validates the model's transfer on real hardware.

\section{Discussion and Limitations}\label{sec:limits}

\textbf{L1 --- headline numbers are model-space.} Phase/scale figures come
from the calibrated simulator; both ends are anchored by wall clock
(\cref{sec:microbench,sec:transfer}) but mid-range transfer error is
2--6\,pp.

\textbf{L2 --- scale of the testbeds.} Serving evidence spans
8-GPU Testbed~A and $2\times8$ Testbed~B (EP16); full DSv3 scale (58 MoE
layers, EP32+) remains extrapolation: per-layer effects are measured,
depth-multiplied ones are projected. Deeper hierarchies
(rail-optimized fabrics, 4+ nodes) may need finer traffic terms than
\cref{sec:comm-topo}.

\textbf{L3 --- the win region is conditional, now mapped.} Small-expert
bf16 shapes show no adaptive-dispatch gain (measured); with fresh
placement and ample replication all policies tie
(\cref{sec:serving-proxy}). The measured win conditions: placement
staleness inside the predicted batch band ($B{=}64$--$256$,
bracketing $B^{*}_{\mathrm{a2a}}\!\approx\!300$, \cref{sec:bstar}), tight replica
budgets, an in-graph token-LP incumbent, long-prefill compute-bound
traffic and Poisson tails inside the win region
(\cref{sec:serving-tb}), and multi-node EP with sufficient expert
compute share (\cref{sec:serving-ep16}). The v1 remap tax
($9$--$20\%$) shrinks to $2$--$11\%$ on v2, scales with MoE depth, and
is paid by the data path, not the solver---\sys{} beats its own noop
control on every 235B configuration.

\textbf{L4 --- static calibration.} Kernel/driver updates require
re-running the ten-minute calibration; the loop is automated and
converges in one iteration.

\textbf{L5 --- the $(G,N)$ model has a floor.} DeepGEMM masked kernels
are faster under \emph{uniform} per-slot loads at equal $(G,N)$: a
third cost dimension that decides model near-ties---why the ensemble
uses a 1\% band, and why the Qwen3 refit turns $b$ negative. A
$G\cdot\max$-slot term is future work; we flag it rather than overfit
to one kernel's quirk.

\textbf{L6 --- the deployed system is not the theoretical optimum, by
design.} The gaps are deliberate trades, not oversights: the problem
is NP-hard, so \texttt{tempo\_fast} is a bounded heuristic (the
certificate of \cref{thm:rr} covers its round-robin core under full
replication; restricted replica sets are an open problem); the solver
plans asynchronously against the previous window's counts (priced at
$5$--$7$\,pp of a $22$\,pp headroom in \cref{sec:staleness},
${\sim}2$\,pp end-to-end); it operates at whole-slot granularity,
inherits EPLB's placement rather than co-optimizing it, and pays a
depth-scaled remap tax (L3) no solver quality removes. Where these
trades bind, the honest reading of our numbers is parity-at-zero-tax
rather than victory---exactly the phase-diagram prescription. Better
implementations of the same principle should widen the win region;
the cost model and the map of when time-based dispatch can and cannot
pay are the durable contributions.

\section{Conclusion}

Balance time, not tokens. A ten-minute black-box calibration exposes
the two-regime structure of expert cost; a phase diagram shows every
fixed proxy has a failure region while the regimes coexist inside
single batches 92--100\% of the time; and a makespan solver over the
calibrated model tracks the per-regime best everywhere at millisecond
cost. On Testbed~B, two flagship models bracket the predicted win
region and the wins replicate; the time model repairs the latency
tail, makes tight replica budgets safe to exploit, and converts a
multi-node locality regression into a gain. Our implementation is one
deployable point, not the optimum of the theory it applies (L6); the
durable contribution is the map. As MoE serving consolidates around
fp8 grouped kernels, large-expert flagships, and multi-node expert
parallelism, dispatching on measured time rather than counted tokens
is both the principled and the profitable choice.

\bibliographystyle{plain}
\bibliography{refs}

\appendix

\section{Proofs}\label{sec:proofs}

\subsection{Theorem~\ref{thm:nphard}}
\emph{Membership.} Given $(x,z)$, verifying makespan $\le T$ is linear
time; the problem is in NP.

\emph{Hardness.} Reduce from Balanced PARTITION: given positive integers
$w_1..w_m$ ($m$ even, $\sum_i w_i = 2W$), decide whether an index set $S$
with $|S|=m/2$ and $\sum_{i\in S} w_i = W$ exists; this is
NP-complete~\cite{gareyjohnson}.

Construct: 2 GPUs; expert $e_i$ with $n_i=w_i$ tokens and
$R(e_i)=\{1,2\}$; $a=c=0$; pick any $T>0$ and set $b=T/(m/2)$,
$\beta=T/W$. Then $t_g\le T$ iff $G_g\le m/2$ \emph{and} $N_g\le W$: the
max form translates one time budget into an activation-cardinality cap
$K=m/2$ and a token-capacity cap $C=W$ per GPU.

($\Leftarrow$) A balanced partition $S$ gives a whole-expert dispatch
with $G_1=G_2=m/2$, $N_1=N_2=W$: makespan exactly $T$.

($\Rightarrow$) Suppose a dispatch has makespan $\le T$ and splits $s\ge
0$ experts. Unsplit experts contribute one activation; split experts
activate on both GPUs, so $G_1+G_2\ge m+s$. But $G_1,G_2\le m/2$ forces
$G_1+G_2\le m$, hence $s=0$: the cardinality budget forbids splitting
outright. The dispatch is a whole-expert bipartition with both sides of
cardinality exactly $m/2$ and token loads $\le W$; since loads sum to
$2W$, both sides carry exactly $W$---a balanced partition. \qed

\emph{Remark.} The difficulty comes from neither degenerate limit---%
$b{=}0$ is an LP (\cref{lem:lp}); $\beta{=}0$ is a polynomial
semi-matching (\cref{lem:semi})---but precisely from their interaction:
the max form makes one makespan budget behave as two knapsack constraints
at once. This mirrors the paper's systems claim that single proxies are
each fine at home and the mixed regime is the hard part.

\subsection{Lemmas~\ref{lem:lp} and~\ref{lem:semi}}
\cref{lem:lp}: with $a=b=0$ the objective is a monotone transform of
$\max_g N_g$; $z$ is irrelevant; the LP relaxation is exact.
\cref{lem:semi}: with $c=\beta=0$, splitting only adds activations, so an
optimal solution activates exactly one replica per live expert; the
problem is the optimal semi-matching of experts to GPUs over replica
edges, solvable in polynomial time via augmenting
paths~\cite{semimatching}. The augmenting-chain stage of
\texttt{tempo\_fast} (\cref{sec:solver}) is its 1/2-step truncation. \qed

\subsection{Ensemble guarantee}
\begin{proposition}\label{prop:ensemble}
Let $\mathcal{C}$ be the candidate set scored by the ensemble: the
local-search output $A_{\mathrm{H}}$, the token-LP dispatch
$A_{\mathrm{LP}}$, and---whenever the replica sets make it
feasible---the round-robin placement $A_3$ of \cref{thm:rr}. With
$M(\cdot)$ the model makespan and $\tau=0.01$, \sys{} outputs
$A^{*}=\arg\min_{A\in\mathcal{C}\setminus\{A_{\mathrm{H}}\}} M(A)$ if
that minimum is below $(1-\tau)M(A_{\mathrm{H}})$, else
$A_{\mathrm{H}}$. Then
$M(A^{*})\le\min\bigl(M(A_{\mathrm{H}}),\,
\min_{A\in\mathcal{C}} M(A)/(1-\tau)\bigr)$: the output is at most a
switching tolerance $\tau$ worse than any candidate under the model,
and never worse than the heuristic. In particular, under full
replication $A_3\in\mathcal{C}$ and
$M(A^{*})\le(\mathrm{OPT}+\max(b,\beta n_{\max}))/(1-\tau)$: the
guarantee of \cref{thm:rr} carries over weakened by the $1/(1-\tau)$
factor, and is recovered exactly by setting $\tau=0$. Under
restricted replication $A_3$ may be infeasible and no such
inheritance is claimed.
\end{proposition}
The tolerance exists because model near-ties are decided by effects outside
$(G,N)$ (\cref{sec:limits}, L5), where measured hardware prefers
$A_{\mathrm{H}}$'s whole-expert structure.

\subsection{Theorem~\ref{thm:rr}}
Sort live experts by tokens, $n_1\ge n_2\ge\dots\ge n_E>0$; write
$\Phi_k=\sum_e\varphi_k(n_e)$ for the total load under measure $k$, and
$\Phi_k(i)$ for GPU $i$'s load under $A_3$.

\emph{Lower bounds.} For every $k$, the total $\Phi_k$ must be carried
regardless of splitting: token mass is conserved, and splitting an
expert only \emph{adds} activations, since each live expert is activated
on at least one GPU. Hence some GPU carries $\ge\Phi_k/g$ and
$\mathrm{OPT}\ge a_k+\Phi_k/g$ for every $k$.

\emph{Round-robin loads.} (i)~GPU $i$ receives
$\lceil(E-i+1)/g\rceil\le\lceil E/g\rceil$ experts. (ii)~\emph{Row-wise
domination}: for $i<j$, GPU $i$'s $r$-th expert is $n_{rg+i}\ge
n_{rg+j}$, and each $\varphi_k$ is nondecreasing, so
$\Phi_k(1)\ge\Phi_k(2)\ge\dots\ge\Phi_k(g)$ \emph{simultaneously for
every $k$}. (iii)~\emph{Head--tail offset}:
\[
\Phi_k(1)=\varphi_k(n_1)+\sum_{r\ge1}\varphi_k(n_{rg+1})
\le\varphi_k(n_1)+\sum_{r\ge1}\varphi_k(n_{rg})
=\varphi_k(n_1)+\Phi_k(g),
\]
using $n_{rg+1}\le n_{rg}$ (descending order). (iv)~By (ii), GPU $g$ is
the minimum under measure $k$, so $\Phi_k(g)\le\Phi_k/g$.

\emph{Combination.} Every GPU $i$ costs
\[
t_i\le\max_k\bigl(a_k+\Phi_k(1)\bigr)
\le\max_k\bigl(a_k+\Phi_k/g+\varphi_k(n_{\max})\bigr)
\le\mathrm{OPT}+\max_k\varphi_k(n_{\max}).
\]
For the two-piece model \eqref{eq:cost}, $\varphi_{\mathrm{act}}(n)=b\,
\mathbf{1}\{n>0\}$ and $\varphi_{\mathrm{tok}}(n)=\beta n$, so the
additive term is $\max(b,\beta n_{\max})=\beta n_{\max}$ whenever a
single hot expert costs more than one activation charge---the only case
in which the bound is not already absorbed by rounding. \qed

\emph{Tightness.} With one expert of $N$ tokens, $\mathrm{OPT}$ splits
it $g$ ways ($\max(a{+}b,\,c{+}\beta N/g)$) while any whole-expert
placement pays $c+\beta N$: the additive term is necessary up to a
$(1{-}1/g)$ factor. Numerically, over a 1296-instance probe suite
($E$ up to 512, $g$ up to 16, including adversarial single-giant
constructions where $M(A_3)/\mathrm{OPT}$ reaches $9.1$) the worst
observed $M(A_3)/(\mathrm{OPT}+\beta n_{\max})$ is $0.998$: the
constant cannot be improved.

\emph{Remark (restricted replication).} Under EPLB-style replica sets,
replace the activation lower bound by the optimal semi-matching value
and $A_3$ by ``fewest-experts replica, token tie-break''; we conjecture
the same additive bound and observe no violation empirically, but the
row-domination argument no longer applies verbatim.

\section{Calibration details}\label{sec:calib-details}
Grid: $G\in\{1..40\}$ $\times$ tokens-per-expert $\in\{1..4096\}$ (log
ladder), three shapes, fp8 masked grouped GEMM and bf16 loop. Timing via
CUDA graphs (30-replay medians) to remove launch noise; expert weight
copies rotated across replays to defeat L2 reuse; \texttt{expected\_m}
fixed per batch size so the DeepGEMM JIT tuner never runs inside timed
regions. Wave-quantization staircases are visible at large $N$ and are
absorbed by the linear piece within fit error on Testbed~A; on Testbed~B
the staircase is explicit in the tile-aware model (\cref{fig:tile}).
Fitting: 2-piece max-affine regression by alternating assignment (50
iterations, least-squares per piece).

\begin{table}[t]
\centering\small
\caption{Fitted cost parameters (Testbed~A, fp8 masked grouped GEMM).}
\label{tab:params}
\begin{tabular}{lrrrr}
\toprule
shape & $b$ ($\mu$s/exp) & $\beta$ ($\mu$s/tok) & $\nstar$ & fit err \\
\midrule
Qwen3-30B  & 1.74  & 0.0108 & ${\approx}161$ & 4.9\% \\
DSv2-Lite  & 2.99  & 0.0179 & ${\approx}168$ & 4.1\% \\
DSv3       & 14.78 & 0.0945 & ${\approx}156$ & 8.0\% \\
DSv3 (bf16)& 45.29 & 0.1370 & ${\approx}331$ & 5.5\% \\
\bottomrule
\end{tabular}
\end{table}

\paragraph{All calibrated parameter sets.}
\cref{tab:allparams} consolidates every parameter set used in the
paper. The two $\beta$ values that appear for the DSv3 shape are not
inconsistent---they are fits to different measurement scopes.
The offline table (\cref{tab:params}) times the grouped-GEMM kernel in
isolation ($\beta{=}0.0945\,\mu$s/token); the EP8 microbenchmark refit
times the \emph{entire per-rank MoE pipeline}, including unfused
per-token kernels (routing, scatter/gather, quantization), giving
$\beta{=}0.358$ and a much larger activation-related constant
($a{:}\,116\to358\,\mu$s from GEMM-only to full pipeline). The larger
full-pipeline constants \emph{strengthen} the two-regime picture: the
flat regime is more pronounced end-to-end than in the kernel alone,
which is why the miscalibration ablation (GEMM-only parameters driving
a full-pipeline dispatch) loses 16\% at $B{=}128$. The Qwen3 full-%
pipeline refit is degenerate ($b<0$, below measurement noise) and is
reported as the negative control of \cref{sec:transfer}, not used.

\begin{table}[t]
\centering\small
\caption{Consolidated calibrated parameters ($\mu$s; per-rank scope
unless noted). ``offline'' = isolated-kernel grid
(\cref{tab:params}); ``mb refit'' = black-box refit on the
8-GPU Testbed~A EP8 microbenchmark's own per-rank log; ``serving'' =
three-term model in the deployed worker; Testbed~B fits are tile-aware
(\cref{eq:tilecost}). The traffic pair $(c_2,\gamma)$ is fitted once,
on the proxy pipeline's static-vs-uniform crossover, and reused
unchanged on Testbed~B; it is the least independently validated part of the
model (\cref{sec:limits}). The two full-pipe rows are the same model
on the $(a,b)$ ridge---see the identifiability analysis in
\cref{sec:transfer}.}
\label{tab:allparams}
\begin{tabular}{llrrrrrrr}
\toprule
context & scope & $a$ & $b$ & $b_2$ & $c$ & $\beta$ & $c_2$ & $\gamma$ \\
\midrule
DSv3, Testbed~A offline & kernel & --- & 14.78 & --- & --- & 0.0945 & --- & --- \\
DSv3, Testbed~A mb refit & GEMM & 116 & 12.99 & --- & 176 & 0.0851 & --- & --- \\
DSv3, Testbed~A mb refit & full pipe & 358 & 16.59 & --- & 266 & 0.358 & --- & --- \\
DSv3, Testbed~A refit it.2 & full pipe & 488 & 5.93 & --- & 294 & 0.356 & --- & --- \\
proxy (DSv2-Lite), Testbed~A serving & 3-term & 16.4 & 2.99 & --- & 15.6 & 0.0179 & 25.0 & 0.10 \\
DSv3, Testbed~B & tile & 37.1 & 8.32 & 2.65 & 23.4 & 0.0336 & 25.0 & 0.10 \\
Qwen3-235B, Testbed~B & tile & 34.9 & 3.95 & 1.23 & 34.8 & 0.0154 & 25.0 & 0.10 \\
\bottomrule
\end{tabular}
\end{table}

\section{Communication model}\label{sec:comm-model}
Per-rank features logged per step: bytes sent/received, source/dest
fan-in/out, measured all-to-all time. A 3-parameter linear model on
received tokens explains dispatch-side variance with 12.7\% makespan
prediction error; 5-parameter variants (send + degree terms) improve fit
marginally and change no policy ranking. Single-node NVLink only;
cross-node $k_r$ recalibration expected (L2).

\section{SGLang integration patch list}\label{sec:patches}
Seven idempotent patch points: (1) dispatch-algorithm registry entry;
(2) solver module install; (3) expert-location metadata init for
\texttt{init\_expert\_location}; (4) recorder gatherer accepting
\texttt{deepep\_mode=auto}; (5) EPLB-rebalance hook re-initializing
\sys{} solvers; (6) relaxed stream-capture mode for the refresher
thread; (7) recorder buffer-size cap. Pitfalls table (DeepEP token caps,
KV-pool sizing under \texttt{--enable-eplb}, NCCL fallback with
\texttt{--disable-custom-all-reduce}) in the artifact (link anonymized
for review).

\section{Additional phase panels}\label{sec:more-panels}
Qwen3-30B and DSv2-Lite shapes, both replication ratios, and regime maps
for all four traces are included in the artifact; rankings match the
main-text panels throughout.

\section{Additional result tables and figures}\label{sec:more-tables}

\begin{figure}[t]
\centering
\includegraphics[width=\linewidth]{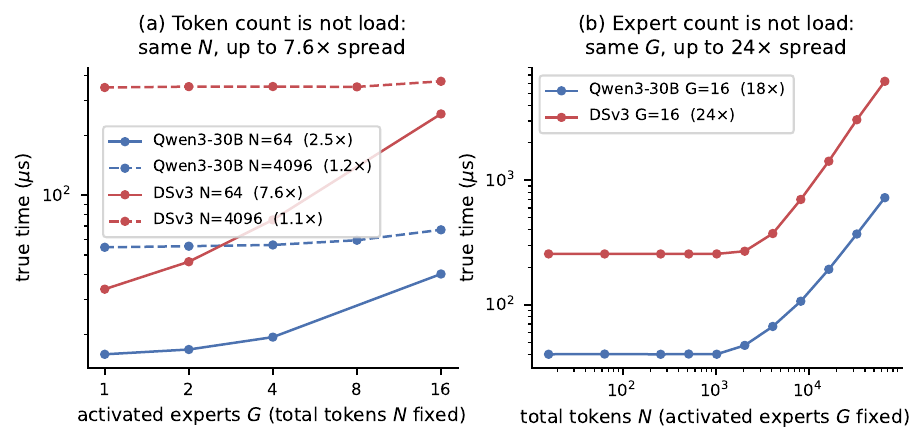}
\caption{Proxy-equivalence sets versus true time on the raw
measurement grid. (a) Same total tokens, different activation counts:
up to $7.6\times$. (b) Same activation count, different token counts:
up to $24\times$.}
\label{fig:proxy}
\end{figure}

\begin{figure}[t]
\centering
\includegraphics[width=\linewidth]{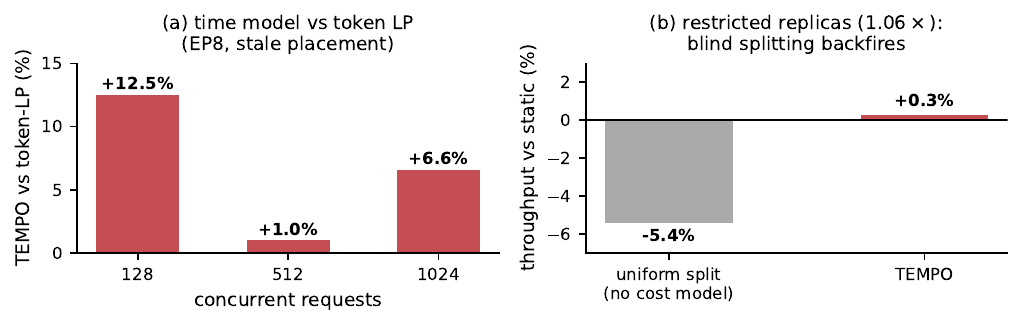}
\caption{Stale-placement serving on the DSv3-shape proxy (EP8, Testbed~A).
(a) \sys{} vs.\ the token-LP dispatcher on the identical data path.
(b) With a tight replica budget ($1.06\times$), uniform splitting
backfires; the cost model recovers parity with static.}
\label{fig:drift}
\end{figure}

\begin{figure}[t]
\centering
\includegraphics[width=\linewidth]{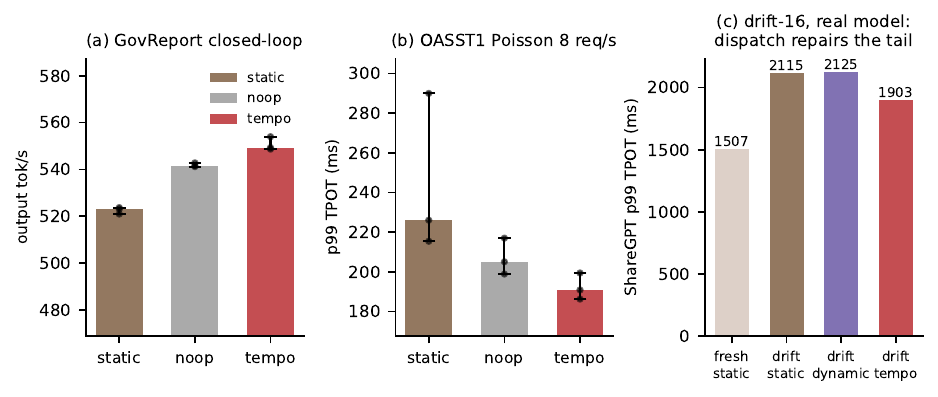}
\caption{Robustness of the Testbed~B win-region claims. (a,b) Three
independent repeat windows per policy (dots; bar = median, whiskers =
range): the GovReport throughput and Poisson-tail wins replicate with
non-overlapping ranges. (c) Drift-16 corruption applied to the real
model: stale placement inflates ShareGPT p99 by $40\%$; \sys{}
recovers a third of the damage while SGLang's dynamic rebalancer
recovers none---and no dispatcher recovers the throughput, which
requires re-provisioning replicas (placement's job).}
\label{fig:robust}
\end{figure}

\begin{figure}[t]
\centering
\includegraphics[width=\linewidth]{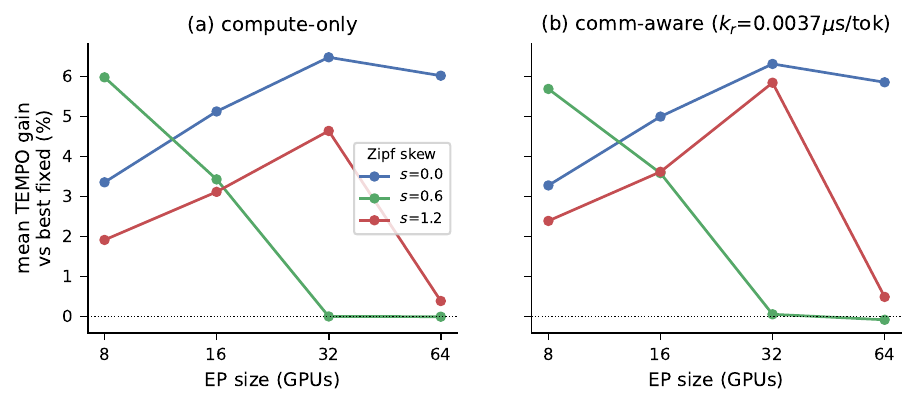}
\caption{Scale extrapolation, DSv3 shape, replication $1.25\times$:
mean \sys{} gain vs.\ best fixed policy across the batch ladder.}
\label{fig:scale}
\end{figure}

\begin{figure}[t]
\centering
\includegraphics[width=\linewidth]{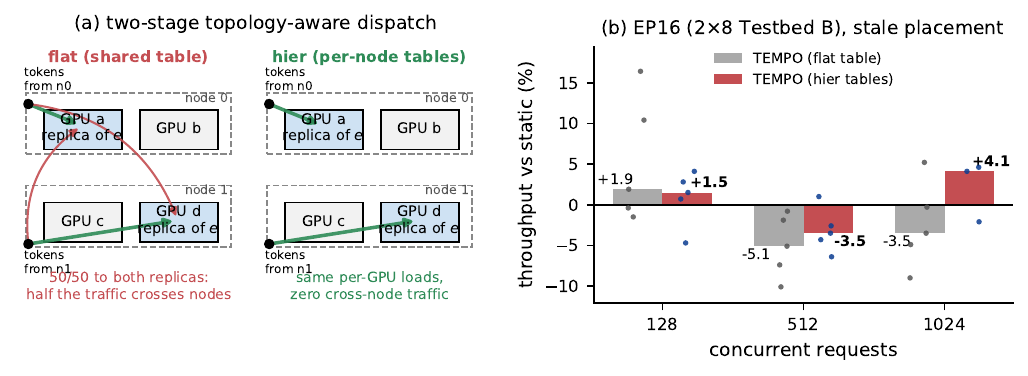}
\caption{Topology-aware dispatch on 2-node EP16 (Testbed~B, stale
placement). (a) The solve fixes per-GPU loads; a same-node-first
transportation split then assigns \emph{sources} to replicas, giving
one table per source node at zero kernel cost. (b) End-to-end: the
flat table loses to static at the all-to-all-bound point; the hier
split flips the sign (median of 5 paired windows, dots are
individual windows).}
\label{fig:hier}
\end{figure}

\begin{table}[t]
\centering\small
\caption{Cumulative stage ablation on the drift-16 scenario: median
makespan/static under the calibrated model (lower is better; $<1$ beats
static). Solve time at $B{=}1024$.}
\label{tab:stageabl}
\begin{tabular}{lrrrr}
\toprule
stages (cumulative) & $B{=}128$ & $B{=}512$ & $B{=}1024$ & ms \\
\midrule
seed only            & 1.000 & 0.883 & 0.869 & 0.3 \\
\ +aug.\ chains      & 0.902 & 0.940 & 0.932 & 0.3 \\
\ +local search      & 0.902 & 0.881 & 0.867 & 0.9 \\
\ +LP ensemble       & 0.902 & 0.880 & 0.865 & 3.5 \\
\midrule
seed+local, no chains & 0.907 & 0.881 & 0.866 & 0.3--2.2 \\
\bottomrule
\end{tabular}
\end{table}

\begin{table}[t]
\centering\small
\caption{Cost-model term ablation on the drift-16 scenario: solve with
terms deleted from the objective, score under the full model (median
makespan/static).}
\label{tab:termabl}
\begin{tabular}{lrrrr}
\toprule
objective & $B{=}128$ & $B{=}512$ & $B{=}1024$ & $B{=}2048$ \\
\midrule
full 3-term \eqref{eq:cost3} & 0.902 & 0.880 & 0.865 & 0.856 \\
drop traffic ($c_2,\gamma$)  & 0.902 & 0.946 & 0.924 & 0.862 \\
drop floor ($a,b$) $\equiv$ tokens only & 0.951 & 0.880 & 0.865 & 0.856 \\
\bottomrule
\end{tabular}
\end{table}

\begin{table}[t]
\centering\small
\caption{Faster-EPLB control (drift16 replay, makespan ratio vs.\
never-refreshed static; median/p95). ``uni'' = EPLB default uniform
split; ``cum16'' = placement rebuilt every 16 windows from cumulative
observed counts; W4 = trailing 4-window refresh.}
\label{tab:refresh}
\begin{tabular}{lcc}
\toprule
policy & $B{=}128$ & $B{=}512$ \\
\midrule
static, never refresh      & 1.00/1.00 & 1.00/1.00 \\
uni, never refresh         & 1.14/1.23 & 1.00/1.05 \\
uni + cum16 refresh        & 1.31/1.44 & 0.85/1.02 \\
static + cum16 refresh     & 1.00/1.15 & 1.03/1.33 \\
\sys{}, never refresh      & 1.00/1.10 & 0.89/0.97 \\
\sys{} + cum16 refresh     & 1.00/1.10 & 0.81/0.95 \\
\sys{} + W4 refresh        & 1.00/1.10 & 0.80/0.93 \\
\bottomrule
\end{tabular}
\end{table}

\begin{table}[t]
\centering\small
\caption{No-solve topology control on Qwen3-235B 2-node EP16, drift-16
placement: output-throughput delta vs.\ same-window static
(round~1/round~2).}
\label{tab:noslv}
\begin{tabular}{lrrr}
\toprule
requests & flat-uniform & hier-uniform & hier-\sys{} \\
\midrule
128  & $-1.2$/$-3.7\%$ & $-3.9$/$-7.0\%$ & $-1.3$/$-5.0\%$ \\
512  & $+0.3$/$-2.2\%$ & $-2.8$/$-1.8\%$ & $+1.1$/$-0.3\%$ \\
1024 & $-3.2$/$-3.4\%$ & $-3.2$/$-1.4\%$ & $-0.8$/$-1.1\%$ \\
\bottomrule
\end{tabular}
\end{table}

\section{Multi-node bring-up notes}\label{sec:bringup}
Three hazards, each worth $10$--$200\times$ in measured performance,
documented for reproducibility. (1)~Raw RDMA was healthy (365\,Gb/s
cross-node \texttt{ib\_write\_bw}) while NCCL hung in transport setup:
the container's HPC-X IBext plugin, auto-loaded by NCCL.
\texttt{NCCL\_NET\_PLUGIN=none} plus explicit HCA/GID selection
restored line rate (decode TPOT $3429\to7.7$\,ms). (2)~The stock NCCL
bundled with PyTorch reached ${\sim}2$\,Gb/s effective allreduce
bandwidth on the same fabric \emph{while reporting} \texttt{NET/IB}
transport---silently crawling; DeepEP's NVSHMEM/IBGDA path on the same
NICs sustained 50\,GB/s, localizing the fault. Preloading the cluster
vendor's NCCL restored $53\,\mu$s allreduce ($100$--$200\times$),
serving TPOT $4.3$\,s$\,\to42$\,ms. Diagnosing collective-transport
health per stack is a prerequisite for any multi-node EP claim.
(3)~fp8 on the new SM architecture: DeepGEMM masked kernels rejected
SGLang's fp32 scales (the architecture requires UE8M0 packing; fixed
by broadening the packing condition), and running LPLB required a
cublasdx fix we upstream---mathdx's bundled cutlass predates the new
architecture, selecting prior-generation float2 FFMA atoms that assert
at launch; defining the float2-math macros for the new architecture
restores it.

\section{MoonEP comparison details}\label{sec:moonep-details}
\emph{Model replay.} We replay MoonEP's policy---prefetch redundant
experts from the current router output so every rank receives exactly
$SK$ tokens---on the recorded Qwen3-235B distributions under the
tile-aware Testbed~B cost model, with two prefetch accountings: \emph{free}
(fully overlapped; optimistic) and \emph{charged} at a $4\times$
HBM-to-interconnect bandwidth ratio per duplicated expert. At decode
batches ($8$--$64$ tokens/GPU) even free prefetch reaches only $-4.8$
to $-10.2\%$ of static makespan---perfect token balance is
activation-blind, and \sys{} takes $-12.1$ to $-16.5\%$ on the same
batches; charged prefetch inverts the sign entirely ($+277$ to
$+320\%$): a duplicate must absorb several tiles' worth of tokens to
amortize its own weight movement, which decode batches never provide.
At prefill batches ($\ge$512 tokens/GPU) all policies converge to
token balance and the free variant exactly matches token-LP---
consistent with MoonEP's training-side claims, where prefetch overlaps
behind long compute.

\emph{Wall clock, real library.} We run MoonEP against DeepEP v2 on
8-GPU Testbed~B with MoonEP's own aligned benchmark (identical routing,
shared timing harness; DSv3 shape, 256 experts, hidden 7168, top-8;
skew swept over MaxVio $0.2$--$20$). At decode scale (128 tokens/rank)
DeepEP dispatch+combine is $95$--$107\,\mu$s and \emph{flat} in skew,
while MoonEP costs $528$--$829\,\mu$s, $66$--$77\%$ of it weight
prefetch ($283$--$502\,\mu$s); even MoonEP's on-GPU planning kernel
($57$--$61\,\mu$s) costs more than DeepEP's entire dispatch. At
prefill scale (8192 tokens/rank) the picture inverts as the phase
diagram predicts: DeepEP degrades with skew ($2.23\to2.85$\,ms as
MaxVio goes $0.2\to20$) while MoonEP stays flat (${\sim}2.55$\,ms),
crossing near MaxVio $\approx 20$.

\end{document}